\documentclass[a4paper,USenglish,cleveref, autoref, thm-restate, numberwithinsect]{lipics-v2021}
\hideLIPIcs
\usepackage{graphicx} 
\usepackage{tikz}
\usetikzlibrary{calc}
\usetikzlibrary{arrows}
\usepackage{amsthm}
\usepackage{todonotes}
\usepackage{mathtools}
\usepackage[algo2e, vlined, ruled, linesnumbered]{algorithm2e}
\usepackage{booktabs}

\renewcommand{\O}{\ensuremath{\mathcal{O}}}
\newcommand{\A}{\ensuremath{\mathcal{A}}}
\newcommand{\G}{\ensuremath{\mathcal{G}}}
\newcommand{\C}{\ensuremath{\mathcal{C}}}
\newcommand{\B}{\ensuremath{\mathcal{B}}}
\newcommand{\F}{\ensuremath{\mathcal{F}}}
\newcommand{\I}{\ensuremath{\mathcal{I}}}
\newcommand{\cf}{\ensuremath{\mathcal{F}}}
\newcommand{\cl}{\ensuremath{\mathcal{L}}}
\renewcommand{\P}{\ensuremath{\mathsf{P}}}
\newcommand{\NP}{\ensuremath{\mathsf{NP}}}
\newcommand{\XP}{\ensuremath{\mathsf{XP}}}
\newcommand{\FPT}{\ensuremath{\mathsf{FPT}}}
\newcommand{\W}{\ensuremath{\mathsf{W[1]}}}
\newcommand{\N}{\ensuremath{\mathbb{N}}}

\newcommand{\Gred}{G_\mathrm{red}}
\newcommand{\Tred}{T_\mathrm{red}}
\newcommand{\sigmakleiner}{\preccurlyeq_\sigma}
\newcommand{\pikleiner}{\preccurlyeq_\pi}
\newcommand{\spikleiner}{\prec_\pi}
\newcommand{\staukleiner}{\prec_\tau}
\newcommand{\taukleiner}{\preccurlyeq_\tau}
\newcommand{\invsig}[1]{\sigma^{\text{-}1}(#1)}
\newcommand{\invpi}[1]{\pi^{\text{-}1}(#1)}
\newcommand{\invtau}[1]{\tau^{\text{-}1}(#1)}
\newcommand{\ssigmakleiner}{\prec_\sigma}
\newcommand{\distD}{\operatorname{d}}
\newcommand{\dist}[3]{\distD_{#1}(#2,#3)}
\newcommand{\start}{\bar v}
\newcommand{\notStart}{{v'}}
\newcommand{\Rrel}[2]{\mathcal{R}(#1,#2)}

\newcommand{\bip}[1]{\operatorname{bip}(#1)}
\newcommand{\searchlabel}{\emph{label}}

\DeclareMathOperator{\umin}{umin}
\DeclareMathOperator{\umax}{umax}

\tikzset{%
    vertex/.style={inner sep=2pt,draw,circle,fill=white}
}

\tikzstyle{treeedge} = [line width=2.5]

\theoremstyle{definition}
\newtheorem{problem}{Problem}

\title{Generalized Graph Search Trees}

\author{Florian Krowiorz}{Institute of Mathematics, Brandenburg University of Technology, Cottbus, Germany}{florian.krowiorz@b-tu.de}{}{}
\author{Robert Scheffler}{Institute of Mathematics, Brandenburg University of Technology, Cottbus, Germany}{robert.scheffler@b-tu.de}{https://orcid.org/0000-0001-6007-4202}{}

\authorrunning{F. Krowiorz and R. Scheffler} 



\nolinenumbers

\begin{document}

\maketitle

\begin{abstract}
Graph search algorithms and their corresponding graph search trees are commonly used in algorithmic graph theory. In recent years, the recognition problem of these graph search trees has received significant attention. So far, the research has focused on two types of search trees: first-in trees that behave like BFS-trees and last-in trees that behave like DFS-trees. The search tree paradigms differ from each other by the parent a vertex is connected to. In first-in trees, it is the first visited neighbor, while in last-in trees it is the last neighbor visited before that vertex. Here, we will generalize these concepts of graph search trees by allowing every preceding neighbor of a vertex to be the parent. We study the complexity of the recognition problem of these generalized graph search trees. We present \NP-completeness proofs for most searches. We also show that the problem is trivial for Generic Search and polynomial-time solvable for several searches on bipartite graphs and chordal graphs. We also study the question how fixing the start vertex influences the complexity of the problem.
\end{abstract}

\section{Introduction}

Graph search algorithms are powerful tools that are used in many graph algorithms as subroutines. The best-known examples of these graph searches are \emph{Breadth First Search} (BFS) and \emph{Depth First Search} (DFS), which have a wide range of applications. Further graph search algorithms are \emph{Lexicographic Breadth First Search} (LBFS)~\cite{rose1976algorithmic}, \emph{Lexicographic Depth First Search} (LDFS)~\cite{corneil2008unified} as well as \emph{Maximum Cardinality Search} (MCS)~\cite{tarjan1984simple}. Although less well known than BFS and DFS, these searches also have several applications. LBFS is a subroutine of recognition algorithms for several graph classes~\cite{bretscher2008simple,corneil2010lbfs,dusart2017new}, LDFS has been used to solve several optimization problems on cocomparability graphs~\cite{corneil2013ldfs,corneil2016power,mertzios2018linear} and MCS can be used to compute minimal triangulations~\cite{berry2004maximum}.

In recent years, graph search algorithms have also become objects of research besides their application as subroutines. Several researchers considered problems to determine whether the vertex orderings computed by these graph searches can fulfill certain properties. An extensively studied example is the \emph{end-vertex problem}~\cite{beisegel2019end,charbithabibmamcarz14,corneil2010,gorzny2017end,kratsch2015end,rong2022graph} that asks whether a search ordering can end in a given vertex $t \in V(G)$ (see \cite[Table~6.1]{gorzny2022related} and \cite[Table~1]{scheffler2023ready} for overviews of known results). Other problems considered the spanning trees produced by those graph searches -- also called \emph{search trees}. In such a tree, every vertex except the start vertex of the search is connected to one of its preceding neighbors. In trees produced by BFS, they are connected to their first visited neighbor, while in DFS trees every vertex is adjacent to the neighbor that was visited last before it. Using a notation introduced by Beisegel et al.~\cite{beisegel2021recognition}, we call a BFS-like search tree \emph{\cf-tree} (short for \emph{first-in tree}) and a DFS-like search tree \emph{\cl-tree} (short for \emph{last-in tree}).

The problem of whether a given spanning tree of a graph is a search tree of a certain search was first considered in the 1980s by Hagerup and Nowak for BFS and DFS~\cite{hagerup1985biconnected,hagerup1985recognition}. The authors showed that both problems can be solved in linear time. In the following years, similar questions have been considered for weighted graphs and parallel algorithms~\cite{korach1993recognition,peng2000recognizing,peng1994recognizing}. In the last few years, these questions have gained new attention. As mentioned above, Beisegel et al.~\cite{beisegel2021recognition} introduced \cf-trees and \cl-trees, a general framework for those questions. They were able to prove that the \cf-tree recognition problem of LBFS, LDFS, and MCS are \NP-complete while the \cl-tree recognition of LDFS can be solved in polynomial time. In the following years, \NP-hardness has also been shown for the \cf-tree recognition of DFS~\cite{scheffler2022recognition}, as well as the \cl-tree recognition of BFS~\cite{scheffler2022recognition}, LBFS~\cite{scheffler2024recognizing}, and Generic Search~\cite{beisegel2024graph} (an overview is given in \cite[Table~2]{scheffler2023ready}). Generic Search is in some sense the most general graph search since it is allowed to take any vertex as the next one as long as one of its neighbors has already been visited. 

Both the \cf-tree recognition problem and the end-vertex problem are special cases of the \emph{partial search order problem} (PSOP)~\cite{scheffler2025partial}. Here, a graph and a partial order on its vertex set are given and one is asked whether there is a search ordering that extends the partial order.

\subparagraph*{Our Contribution} Although Beisegel et al.~\cite{beisegel2021recognition} only considered the recognition of \cf-trees and \cl-trees, they also introduced \emph{generic search trees} of a search ordering as those spanning trees that have an edge from every vertex except the start vertex to exactly one of its preceding neighbors. To avoid confusions with search trees constructed by Generic Search, we will call these trees here \emph{generalized search trees} (or \emph{\G-trees} for short). Similar to the recognition problem of \cf-trees and \cl-trees, we will consider the \G-tree recognition problem where a graph $G$ and a spanning tree $T$ are given and one has to decide whether there exists a search ordering where $T$ is a \G-tree of $G$. The rooted variant of this problem, where the start vertex of the graph search is fixed, is a special case of the PSOP for partial orders whose Hasse diagram forms a spanning tree of the given graph rooted in a minimal element of the partial order.

We show that the \G-tree recognition problem is trivial for Generic Search since every spanning tree is a \G-tree of some Generic Search ordering. Contrary, we show that for (L)BFS, (L)DFS, MCS, and MNS the problem is \NP-complete. Furthermore, we study the relationship between the complexity of the rooted and the unrooted problem. While it is easy to see that the unrooted problem can be solved using an algorithm for the rooted problem, the reverse is not true in general. We give a sufficient condition for a graph search and a graph class such that both problems are polynomial-time equivalent. Besides this, we give polynomial-time algorithms for some of the searches for graph classes. We solve the problem for layered searches such as BFS and LBFS on bipartite graphs. Furthermore, we show that a large family of graph searches, called strictly comp-monotone searches, has the same \G-trees on chordal graphs. As this family includes LBFS, LDFS, MCS and MNS, we can use the polynomial-time algorithms given for the PSOP of MCS on this class~\cite{rong2026partial,scheffler2025partial} to solve the the \G-tree problem for all these searches. An overview of the results is given in \cref{tab:results-tree}.

Note that many results of this paper are based on the Bachelor's thesis of the first author~\cite{krowiorz2023algorithmische}. Some of the results given in \cref{sec:chordal} are contained in the PhD thesis of the second author~\cite{scheffler2023ready}. 

\begin{table}[t]
	\centering
	\caption{Complexity of the \G-tree recognition problem of different searches on different graph classes. T~stands for trivial problems where the answer is always yes. L~stands for linear-time algorithms, P~for polynomial-time algorithms and NPC for \NP-complete. The result marked with an asterisk is only proven for the rooted case where the start vertex of the search is fixed. }\label{tab:results-tree}
    \small
	\begin{tabular}{c c c c c c c c}
		\addlinespace
		\toprule
		$\G$-Tree Results   & GS & BFS                 & LBFS           & DFS\phantom{\textsuperscript{*}}  & LDFS           & MCS            & MNS        \\[0.8ex] \toprule
		All Graphs     & 
		T  & 
		NPC & 
		NPC & 
		NPC\phantom{\textsuperscript{*}} & 
		NPC & 
		NPC & 
		NPC \\ 
		Weakly Chordal & 
		T & 
		  ? & 
		? & 
		NPC\phantom{\textsuperscript{*}} & 
        NPC & 
		  NPC  & 
		NPC \\ 
		Chordal & 
		T & 
		? & 
		P & 
		NPC\phantom{\textsuperscript{*}} & 
		P & 
		P & 
		P \\ 
		Split & 
		T & 
		P & 
		P & 
		NPC{\textsuperscript{*}} & 
		P & 
		P & 
		P \\ 
		Bipartite & 
		T & 
		P & 
		P & 
		NPC\phantom{\textsuperscript{*}} & 
		? & 
		? & 
		? \\ 
		Chordal Bipartite & 
		T & 
		P & 
		P & 
		NPC\phantom{\textsuperscript{*}} & 
		? & 
		? & 
		?  \\
        Complete Bipartite & 
		T & 
		L & 
		L & 
		T\phantom{\textsuperscript{*}} & 
		T & 
		T & 
		T  \\[0.8ex]
		\bottomrule \addlinespace
	\end{tabular}
\end{table}

\section{Preliminaries}
\subparagraph{General Notation}

The graphs considered in this paper are finite, undirected, simple and connected. Given a graph $G$, we denote by $V(G)$ the \emph{set of vertices} and by $E(G)$ the \emph{set of edges}. The terms $n(G)$ and $m(G)$ describe the number of vertices and edges of $G$, respectively, i.e., $n(G) = |V(G)|$ and $m(G) = |E(G)|$. For a vertex $v\in V(G)$, we denote by $N_G(v)$ the \emph{(open) neighborhood} of $v$ in $G$, i.e., the set $N_G(v)=\{u\in V(G)\mid uv\in E(G)\}$ where $uv$ denotes an edge between $u$ and $v$. Given a set $S \subseteq V(G)$, the graph $G[S]$ is the \emph{subgraph of $G$ induced by $S$}.

The \emph{distance} $\dist G v w $ of two vertices $v$ and $w$ in $G$ is the length (i.e., the number of edges) of the shortest $v$-$w$-path in $G$. For some $\ell \in \N$, we define $N_G^\ell(v)$ as the set containing all vertices whose distance to the vertex $v$ in $G$ is equal to $\ell$.

A \emph{vertex ordering} of $G$ is a bijection $\sigma: V(G) \to \{1,2,\dots,|V(G)|\}$. We denote by $\sigma(v)$ the position of vertex $v\in V(G)$. Given two vertices $u$ and $v$ in $G$, we say that $u$ is \emph{to the left} (resp. \emph{to the right}) of $v$ if $\sigma(u)<\sigma(v)$ (resp. $\sigma(u)>\sigma(v)$) and we denote this by $u \prec_{\sigma}v$ (resp.) $u \succ_{\sigma}v$). We write $u \sigmakleiner v$ if $u \ssigmakleiner v$ or $u = v$ holds. Depending on the context, we also use $\sigmakleiner$ to denote the set of all tuples $(u,v)$ such that $u \sigmakleiner v$.
For any nonempty subset $W \subseteq V$ of vertices we denote with $\sigma_{\vert W}$ the \emph{restriction} of $\sigma$ to the elements of $W$, i.e., $\sigma_{\vert W}(v) \coloneqq \vert \{w \in W \mid w \sigmakleiner v\}\vert$.

A subgraph $H$ of $G$ is \emph{spanning} if $V(G) = V(H)$. A \emph{clique} in a graph $G$ is a set of pairwise adjacent vertices and an \emph{independent set} in $G$ is a set of pairwise non-adjacent vertices. A vertex $v$ is \emph{simplicial} if its neighborhood induces a clique. A vertex $v$ of a connected graph $G$ is a \emph{cut vertex} if $G - v$ is not connected.

The \emph{union of two graphs} $G \cup H$ is the graph that contains all vertices and edges of $G$ and of $H$, i.e., $V(G \cup H) = V(G) \cup V(H)$ and $E(G \cup H) = E(G) \cup E(H)$. A \emph{cut vertex decomposition} of a graph $G$ is a three-tuple $(A,B,v)$ where $A$ and $B$ are subgraphs of $G$, $A \cup B = G$ and $V(A) \cap V(B) = \{v\}$.

A graph is \emph{bipartite} if its vertex set can be partitioned into two independent sets $X$ and~$Y$. A vertex of a bipartite graph is called \emph{biuniversal} if it is completely adjacent to one of the two independent sets. A bipartite graph is \emph{complete bipartite} if all vertices are biuniversal. A graph is \emph{weakly chordal} if neither $G$ nor its complement contains an induced cycle of length $\geq 5$. A graph is \emph{chordal} if it does not contain an induced cycle of length $\geq 4$. A vertex ordering $\sigma$ of a graph $G$ is a \emph{perfect elimination ordering} if every vertex $v$ is simplicial in the graph $G[S(v)]$ with $S(v) := \{w \mid w \preceq_\sigma v\}$. A graph $G$ has a PEO if and only if $G$ is chordal~\cite{rose1970triangulated}. A \emph{split graph} $G$ is a graph whose vertex set can be partitioned into sets $C$ and $I$ such that $C$ is a clique in $G$ and $I$ is an independent set in $G$. It is easy to see that every split graph is chordal. 

A \emph{tree} is an acyclic connected graph. A \emph{spanning tree} of a graph $G$ is an acyclic connected spanning subgraph of $G$. A tree together with a distinguished \emph{root vertex} $r$ is said to be \emph{rooted}. In such a rooted tree $T$, a vertex $v$ is an \emph{ancestor} of vertex $w$ if $v$ lies on the unique path from $r$ to $w$. If $v$ is also adjacent to $w$, then $v$ is called the \emph{parent} of $v$ in $T$. A vertex $w$ is called the \emph{child} of $ v $ if $ v $ is the parent of $ w $. The ancestor-relation describes a partial order $\Rrel{T}{r}$ on the vertex set of $T$, i.e., $\Rrel{T}{r} := \{(u,v) \in V(T) \times V(T) \mid \text{$u$ is an ancestor of $v$}\}$.

\subparagraph{Graph Searches}
There are different ways to define what a graph search is. The authors of  ~\cite{beisegel2021recognition} and~\cite{corneil2016tie}, for example, define a graph search to be an algorithm that, given a graph $G$ as input, produces some vertex ordering of $G$ as output. Here, we use a slightly more general definition that omits the algorithmic part. Note that this definition was already used in~\cite{scheffler2025semi,scheffler2025leaves}. We say that a \emph{graph search} $\A$ is a function that maps every graph $G$ to a set $\A(G)$ of vertex orderings of $G$. The elements of the set $\A(G)$ are the \emph{$\A$-orderings of $G$}. Every prefix of an $\A$-ordering is called an \emph{\A-prefix}. We denote by $\A(G,v)$ the orderings of $\A(G)$ that start with vertex $v \in V(G)$. We call a graph search \emph{natural} if for each graph $G$ and each vertex $v \in V(G)$ the set $\A(G,v)$ is not empty.

Similarly to what has been done in \cite{scheffler2025partial,scheffler2025leaves}, we introduce most of the graph searches considered here using the meta-algorithm Label Search($\prec_\A$) given in \cref{algo:ls}. This algorithm is an adaptation of a framework introduced by Corneil~et~al.~\cite{corneil2016tie}. In this algorithm, every vertex $v \in V(G)$ has a set $\searchlabel(v) \subseteq \{1,\dots,n(G)\}$ that contains the indices of the visited neighbors of $v$ in the ordering. The overall idea of the meta-algorithm follows from the fact that the ordering of the vertices that have been already visited by some graph search implies some ordering constraints on the unvisited vertices. To describe these constraints, the framework of LabelSearch($\prec_\A$) uses a strict partial order $\prec_{\cal A}$ on the finite subsets of $\N^+$. Fixing this partial order $\prec_{\cal A}$ also fixes the graph search $\A$, i.e., the set $\A(G)$ of a graph $G$ contains exactly those vertex orderings of $G$ that can be computed by Label Search($\prec_\A$).

\begin{algorithm2e}[t]
\small
    \KwIn{A graph $G$} 
    \KwOut{A search ordering $\sigma$ of $G$}
    \Begin{
		\lForEach{$v \in V(G)$}{\searchlabel($v$) $\leftarrow$ $\emptyset$}
		\For{$i$ $\leftarrow$ $1$ \KwTo $n(G)$}{
			\emph{Eligible} $\leftarrow$ $\{x \in V(G) \mid x$ unvisited and $\nexists$ unvisited $y \in V(G)$ \\ \mbox{}\phantom{\emph{Eligible} $\leftarrow$ $\{x \in V(G) \mid $}such that \searchlabel($x$) $\prec_{\cal A}$ \searchlabel($y$)$\}$\;
			let $v$ be an arbitrary vertex in \emph{Eligible}\;\label{line:ls}
			$\sigma(v)$ $\leftarrow$ $i$\tcc*{assigns to $v$ the number $i$}
			\lForEach{unvisited vertex $w \in N(v)$}{\searchlabel($w$) $\leftarrow$ \searchlabel($w$) $\cup$ $\{i\}$}
		}
	}
    \caption{Label Search($\prec_{\cal A}$)}\label{algo:ls}
\end{algorithm2e}

In the following, we will present the strict partial orders that define the graph searches considered in this paper (see~\cite{corneil2016tie}). In all these definitions, we assume $A$ and $B$ to be finite subsets of $\N^+$. We define $\umin(A)$ as $\infty$ if $A = \emptyset$ and as the minimal element of $A$, otherwise. We define $\umax(A)$ as $0$ if $A = \emptyset$ and as maximal element of $A$, otherwise.

The \emph{Generic Search} (GS) is equal to the Label Search($\prec_{GS}$) where $A \prec_{GS} B$ if and only if $A = \emptyset$ and $B \neq \emptyset$. Thus, any vertex with a visited neighbor can be visited next.

The partial label order $\prec_{BFS}$ for \emph{Breadth First Search} (BFS) is defined as follows: $A \prec_{BFS} B$ if and only if $\umin(A) > \umin (B)$. For the \emph{Lexicographic Breadth First Search} (LBFS)~\cite{rose1976algorithmic}, we consider the partial order $\prec_{LBFS}$ with $A \prec_{LBFS} B$ if and only if $\umin(A \setminus B) > \umin(B \setminus A)$. Both BFS and LBFS are \emph{layered}, i.e., a search ordering starting with vertex $r$ traverses the sets $N^\ell_G(r)$ in ascending order of the distance values $\ell$. We sometimes use the term \emph{layer} if we refer to a set $N^\ell_G(r)$ for some $\ell \in \N$.

The partial label order $\prec_{DFS}$ for \emph{Depth First Search} (DFS) is defined as follows: $A \prec_{DFS} B$ if and only if $\umax(A) < \umax (B)$. For the \emph{Lexicographic Depth First Search}~\cite{corneil2008unified} we use the strict partial order $\prec_{LDFS}$ where $A \prec_{LDFS} B$ if and only if $\umax(A \setminus B) < \umax(B \setminus A)$.

For \emph{Maximum Cardinality Search} (MCS)~\cite{tarjan1984simple} we define $\prec_{MCS}$ with $A \prec_{MCS} B$ if and only if $|A| < |B|$. The \emph{Maximal Neighborhood Search} (MNS)~\cite{corneil2008unified} uses $\prec_{MNS}$ with $A \prec_{MNS} B$ if and only if $A \subsetneq B$. The definition of that partial orders implies that any LBFS ordering is a BFS ordering, any LDFS ordering is a DFS ordering, any LBFS, LDFS, and MCS ordering is an MNS ordering and the orderings of all presented searches are GS orderings (see \cite[Theorem 5.6]{corneil2016tie}).

If we apply a graph search to a graph, then at any point of the computation more than one vertex could have a maximal label. For some results in this paper, we need to break those ties. To this end, $\A^+$-searches have been introduced. Given a graph search $\A$ and a vertex ordering $\rho$, the $\A^+(G,\rho)$ ordering of $G$ is the ordering $\sigma = (v_1, \dots, v_n)$ where for each $i \in \{1,\dots,n-1\}$ the vertex $v_{i+1}$ is the leftmost vertex in $\rho$ among all vertices $x$ for which $(v_1, \dots, v_i, x)$ is a prefix of an $\A$-ordering of $G$. Note that $\A^+(G,\rho)$ is well-defined as long as $\A(G)$ is not empty.

\section{G-Trees}\label{sec:g-trees}

We start with the definition of the main structure considered in this paper.

\begin{definition}[$\G$-tree]
    Let $G$ be graph and $T$ be a spanning subgraph of $G$. Let $\sigma$ be a vertex ordering of $G$. We call $T$ a \emph{$\G$-tree} of $G$ regarding $\sigma$ if for every vertex $v \in V(G)$ it holds that $\sigma(v) = 1$ or there exists exactly one $w \in N_T(v)$ with $w \ssigmakleiner v$. We will denote the set of all subgraphs of $G$ that have the property described above as $\G(G,\sigma)$.
\end{definition}

Next, we give some alternative characterizations of $\G$-trees. In particular, we show that $\G$-trees are spanning trees.

\begin{lemma}
\label{lemma:char-g-trees}
    Let $G$ be a graph and $T$ be a spanning subgraph of $G$. Let $\sigma$ be a vertex ordering of $G$. The following statements are equivalent.
    \begin{enumerate}[(i)]
        \item $T\in \G(G,\sigma)$
        \item$T[\{\invsig{1},\dots,\invsig{i}\}]$ is a tree for every $i \in \{1,\dots,n(G)\}$.
        \item $T$ is a spanning tree of $G$ and $\Rrel{T}{\invsig{1}}$ is a subset of $\sigmakleiner$.\label{prop:char3}
    \end{enumerate}
\end{lemma}

\begin{proof}
    For better readability, we will denote $T[\{\invsig{1},\dots,\invsig{i}\}]$ by $T_i$ and $\invsig{i}$ by $v_i$ for every $i \in \{1,\dots,n(G)\}$.
    
    $(i)\Rightarrow (ii):$ Suppose that $T$ is a $\G$-tree of $G$ regarding $\sigma$. $T_1$ is a graph with exactly one vertex and, thus, a tree. Suppose that for some $i \in \{1,\dots,n(G)-1\}$ the graph $T_i$ is a tree. Since $T$ is a $\G$-tree, we know that there is exactly one vertex $w \in V(G)$ adjacent to $v_{i+1}$ which is to the left of $v_{i+1}$ in $\sigma$. Thus, $T_{i+1}$ arises from $T_i$ by adding $v_{i+1}$ as a leaf. Since the class of trees is closed under the addition of leaves, it follows that $T_{i+1}$ is also a tree.

$(ii)  \Rightarrow (iii):$ Suppose, $T_i$ is a tree for every $i \in \{1,\dots,n(G)\}$. Since $V(T) = V(G)$ and $T_{n(G)} = T$, the graph $T$ is a spanning tree of $G$. It remains to show that for every $(u,w) \in \Rrel T {v_1}$ it holds that $u \sigmakleiner w$. Assume for contradiction that there is a tuple $(u,w) \in \Rrel T {v_1}$ such that $w \ssigmakleiner u$. Since $T_{\sigma(w)}$ is a tree, there must be a $v_1$-$w$-path $P$ within $T_{\sigma(w)}$. Since $u$ is to the right of $w$ in $\sigma$, $u$ is not in $T_{\sigma(w)}$ and, thus, not on $P$. As $T_{\sigma(w)}$ is a subgraph of $T$, $P$ is also a path in $T$. Since $(u,w) \in \Rrel T {v_1}$, there is a $v_1$-$w$-path in $T$ that contains $u$. It follows that there are two different $v_1$-$w$-paths in $T$. This contradicts the fact that $T$ is a tree. Therefore, such a tuple cannot exist and thus $\Rrel T {v_1} \subseteq \sigmakleiner$ holds true.

$(iii)  \Rightarrow (i):$ Suppose, $T$ is a spanning tree of $G$ and $\Rrel{T}{v_1}$ is a subset of $\sigmakleiner$. Let $v$ be a vertex of $G$ with $\sigma (v) \neq 1$. Since $T$ is a spanning tree of $G$, there is a unique path $P$ in $T$ from $v_1$ to $v$. Let $w$ be the predecessor of $v$ in $P$. This implies that $(w,v) \in \Rrel T {v_1}$. Since $\Rrel T {v_1}$ is a subset of $\sigmakleiner$, it holds that $w \ssigmakleiner v$. If $w'$ is a neighbor of $v$ in $T$ with $w' \neq w$, then $P.w'$ is the unique $v_1$-$w'$-path in $T$. Using the same argument as above, we conclude that $v \ssigmakleiner w'$. Thus, $w$ is the only neighbor of $v$ in $T$ with $w \ssigmakleiner v$. 
 \end{proof}

Similarly as was done for \cf-trees and \cl-trees in~\cite{beisegel2021recognition}, we introduce the $\G$-tree recognition problem of a graph search $\A$.

\begin{problem}[$\G$-tree recognition problem of graph search $\A$]~
    \begin{description}
        \item[Instance:] Graph $G$, spanning tree $T$ of $G$
        \item[Question:] Is there a vertex ordering $\sigma \in \A(G)$ such that $T \in \G(G,\sigma)$
    \end{description}
\end{problem}

We will also consider the \emph{rooted $\G$-tree recognition problem} where we have given a vertex $\start \in V(G)$ as additional input and ask whether there is an ordering $\sigma \in \A(G,\start)$ with $T \in \G(G,\sigma)$. Note that the third property of \cref{lemma:char-g-trees} implies that the rooted $\G$-tree recognition problem is a special case of the \emph{Partial Search Order Problem}~\cite{scheffler2025partial}, where one is given a graph and a partial order on the graph's vertex set and asks whether there is a search ordering that is a linear extension of the partial order. For the rooted $\G$-tree recognition problem, the partial order's Hasse diagram is a tree rooted in the unique minimal element. Trotter~\cite{trotter1992combinatorics} calls such partial orders \emph{cs-trees} (short for \emph{computer science trees}). Even more restrictive, the edges of the Hasse diagram of the partial order are edges in the given graph.

Using the characterization given in \cref{lemma:char-g-trees}, we can conclude that every spanning tree is a \G-tree of Generic Search for every possible root vertex, i.e., the (rooted) \G-tree recognition problem of Generic Search is trivial.

\begin{theorem}\label{thm:gs}
    Let $G = (V,E)$ be a graph, $T$ be a spanning tree of $G$ and $\start \in V$. Then there exists a Generic Search ordering $\sigma$ of $G$ starting with $\start$ such that $T \in \G(G,\sigma)$.
\end{theorem}
\begin{proof}
    Let $\sigma$ be any linear extension of $\Rrel{T}{\start}$. Property~(\ref{prop:char3}) of \cref{lemma:char-g-trees} implies that $T$ is a $\G$-tree of $G$ regarding $\sigma$. This choice of $\sigma$ also guarantees that $\sigma(\start) = 1$. It remains to show that $\sigma$ is also a Generic Search ordering of $G$. As $\sigma$ is a linear extension of $\Rrel{T}{\start}$, for every vertex $w \in V(G) \setminus \{\start\}$ its parent $u$ in $T$ is to left of $w$ in $\sigma$. Thus, every vertex but $\start$ has a neighbor to its left in $\sigma$ and, hence, $\sigma$ is a Generic Search ordering.
\end{proof}

\subparagraph{Comparing the complexity of the rooted and unrooted problem}

It is easy to see that a polynomial-time algorithm for the rooted $\G$-tree recognition problem of a graph search $\A$ implies a polynomial-time algorithm for the unrooted problem since we can apply the algorithm for every possible root. A natural question is under which conditions the inverse implication does also hold. It is not difficult to construct a graph search and a graph class for which the inverse implication is not true, unless $\P = \NP$. 

\begin{example}\label{exm:search}
    We define the graph search $\A_1$ using the meta algorithm given in \cref{algo:ls}. Let $A \prec_{\A_1} B$ if and only if $1 \notin A \cap B$ and $\umax(A) < \umax(B)$.
    In other words, the search works like DFS with the difference that the neighbors of the start vertex are incomparable. 

    Now let $\C$ be the class of graphs that contain a universal vertex, i.e., the domination number is equal to one. The \G-tree recognition problem is trivial on $\C$ as every vertex ordering starting in the universal vertex is an $\A_1$-ordering and, thus, we can simply use some Generic Search ordering of the given tree starting in the universal vertex. 
    
    Contrarily, the rooted $\G$-tree recognition problem of $\A_1$ on class $\C$ is as least as hard as the rooted \G-tree recognition problem of DFS on general graphs. For a reduction, let $(G,T,\start)$ be the input of the rooted \G-tree problem of DFS. We construct graph $G'$ by adding a universal vertex $u$ to $G$ and appending a leaf $w$ to $u$. We construct $T'$ by adding the edges $wu$ and $u\start$ to $T$. Solving the rooted \G-tree recognition problem of $\A_1$ for $(G',T',w)$ is equivalent to solving the rooted \G-tree recognition problem of DFS for $(G,T,\start)$. Note that we will prove in \cref{sec:dfs} that the rooted \G-tree recognition problem of DFS is \NP-complete.
\end{example}

Here, we present the family of self-similar searches and we will show that the rooted and the unrooted $\G$-tree recognition problem of such searches are polynomial-time equivalent on general graphs and on graph classes fulfilling certain conditions.

\begin{definition}
    \label{def:self-similar}
    We call a graph search $\A$ \emph{self-similar} on a connected graph $G = (V,E)$ if the following statements hold true for every cut vertex decomposition $(Z,Z',\start)$ of $G$:
    
    \begin{enumerate}[(i)]
        \item $\forall \sigma \in \A(Z,\start) :\quad \forall \tau \in \A(Z'): \quad \exists \pi \in \A(G):\quad  \sigmakleiner \subseteq \pikleiner\quad \land \quad \taukleiner \subseteq \pikleiner$
        \item $\forall \pi \in \A(G) :\quad  \pi^{\text{-} 1}(1) \in V( {Z'})\quad \Rightarrow \quad \pi_{\vert V({Z})} \in \A(Z,\start)$ 
    \end{enumerate}
       
    We say $\A$ is \emph{self-similar} on a graph class $\C$ if it is self-similar on every connected graph of $\C$. We call a graph search $\A$ \emph{self-similar} if it is self-similar on every connected graph $G$.
\end{definition}

The first property says that every search ordering of $Z$ starting in vertex $\start$ can be combined with any search ordering of $Z'$. This combination can be arbitrary, i.e., the orderings are not necessarily be concatenated but can be shuffled into each other. The second conditions states that any search ordering of the whole graph that starts in $Z'$ contains a subordering of $Z$ that is a search ordering starting in $\start$.

It can easily be checked that GS, (L)BFS, (L)DFS, MCS, and MNS are self-similar. This does not hold for all instances of Label Search since the search $\A_1$ given in \cref{exm:search} is not self-similar (see \cref{fig:self-similar}). 
The next theorem identifies some sufficient properties of $\prec_\A$ which ensure that LabelSearch$(\prec_\A)$ is self-similar.

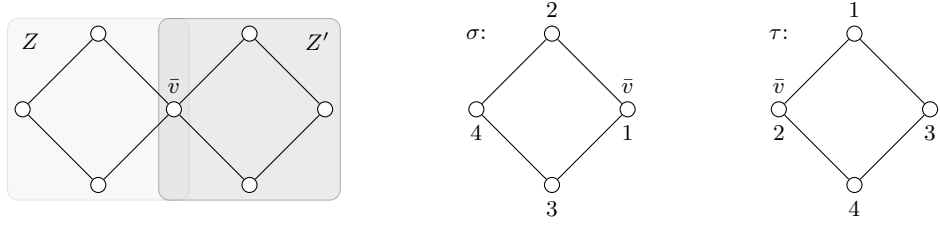
\begin{figure}
    \centering
    \begin{tikzpicture}
        \footnotesize

        \draw[fill=lightgray, opacity=0.1, rounded corners] (-2.2,-1.2) rectangle (0.2,1.2);
        \draw[fill=lightgray, opacity=0.3, rounded corners] (-0.2,-1.2) rectangle (2.2,1.2);
        \node at (-1.9,0.9) {$Z$};
        \node at (1.9,0.9) {$Z'$};
        \node[vertex, label=90:$\start$] (0) at (0,0) {};
        \node[vertex] (-1) at (-1,1) {};
        \node[vertex] (-2) at (-2,0) {};
        \node[vertex] (-3) at (-1,-1) {};
        \node[vertex] (1) at (1,1) {};
        \node[vertex] (2) at (2,0) {};
        \node[vertex] (3) at (1,-1) {};

        \draw (0) -- (-1) --(-2) -- (-3) -- (0) -- (1) -- (2) -- (3) -- (0);

        \begin{scope}[xshift=6cm]
            \node at (-2,1) {$\sigma$:};
            \node[vertex, label=90:$\start$, label=-90:$1$] (0) at (0,0) {};
            \node[vertex, label=90:$2$] (-1) at (-1,1) {};
            \node[vertex, label=-90:$4$] (-2) at (-2,0) {};
            \node[vertex, label=-90:$3$] (-3) at (-1,-1) {};

            \draw (0) -- (-1) --(-2) -- (-3) -- (0);
        \end{scope}

        \begin{scope}[xshift=8cm]
            \node at (0,1) {$\tau$:};
            \node[vertex, label=90:$\start$, label=-90:$2$] (0) at (0,0) {};
            \node[vertex, label=90:$1$] (1) at (1,1) {};
            \node[vertex, label=-90:$3$] (2) at (2,0) {};
            \node[vertex, label=-90:$4$] (3) at (1,-1) {};

            \draw (0) -- (1) -- (2) -- (3) -- (0);
        \end{scope}
    \end{tikzpicture}
    \caption{Example that shows that the search $\A_1$ is not self-similar as it does not fulfill condition~(i) of \cref{def:self-similar}. The orderings $\sigma$ and $\tau$ are $\A_1$-orderings of $Z$ and $Z'$, respectively. An $\A_1$-ordering of $G$ that contains $\sigma$ and $\tau$ as subordering has to start with the start vertex of $\tau$. However, then the search has to follow an DFS ordering on $Z$ which does not fit to $\sigma$.}
    \label{fig:self-similar}
\end{figure}

\begin{theorem}\label{thm:self-similar-label}
    Let $\prec_\A$ be a partial order of the finite subsets of $\N^+$. We define $\A$ to be the graph search LabelSearch($\prec_\A$). $\A$ is self-similar if
    \begin{enumerate}[(a)]
        \item \label{assumption:self-similarity-a}
        for each finite nonempty label $A$ it holds that 
        $\emptyset \prec_\A A$;
        \item \label{assumption:self-similarity-b}
        for all labels $A$ and $B$ and every non-decreasing function $\operatorname f \colon \N^+ \rightarrow \N^+$ it holds that
        \[A \prec_\A B \quad \Leftrightarrow \quad \operatorname f[A] \prec_\A \operatorname f[B];\]
        \item \label{assumption:self-similarity-c}
        for all labels $A$, $B$, $C$ it holds that
        if $A$ and $B$ are incomparable due to $\prec_\A$ and $B \prec_\A C$, then $A \prec_\A C$.
    \end{enumerate}
\end{theorem}

\begin{proof}
    \newcommand{\preNeighbors}[4]{N^{#1}_{#2}(#3,#4)}
    Throughout this proof, we use the following notation: For every graph $G$, vertex ordering $\sigma$ of $G$ and $u,v \in V(G)$ let $\preNeighbors \sigma G u v$ be the set  \[\preNeighbors \sigma G u v \coloneqq \{ \sigma(w) \mid w \in N_G(u),\; w \sigmakleiner v\}.\]

    Note that if $\sigma$ was generated by a label search algorithm then $\preNeighbors \sigma G u v$ is exactly the label of $u$ in the iteration in which $v$ was visited. It is easy to see that $\sigma$ is an element of $\A(G)$ if and only if
    \begin{equation}
    \label{eq:criterium_label_search}
        \forall u,v \in V(G) \colon \quad u \ssigmakleiner v \quad \Rightarrow \quad \preNeighbors \sigma G u u \not \prec_\A \preNeighbors \sigma G v u.
    \end{equation}
    For details, we refer to \cite[Property 3.1]{corneil2016tie}. From now, let $G = (V,E)$ be an arbitrary graph with cut vertex composition $(Z,Z', \start)$. \\

    \proofsubparagraph{Proof of Property (i) of \Cref{def:self-similar}} Let $\sigma$ be an $\A$-ordering of $Z$ starting in $\start$ and $\tau$ be an $\A$-ordering of $Z'$. We have to prove the existence of an ordering $\pi \in  \A(G)$ such that $\pi$ is an extension of $\sigma$ an $\tau$. We claim that $\pi \coloneqq \A^+(G,\tau\sigma)$ satisfies this condition, where
    \[\tau\sigma \coloneqq (\invtau 1, \dots , \invtau {\vert Z' \vert}, \invsig 2,\dots ,\invsig {\vert Z \vert} ). \]
    
    Since $\pi$ is by definition an $\A$ ordering of $G$ it suffices to verify that $\pi$ is an extension of $\sigma$ and $\tau$. We do this by arguing inductively that the label search algorithm with the given tie breaker chooses in each iteration either the leftmost unvisited vertex of $Z$ in $\sigma$ or the leftmost unvisited vertex of $Z'$ in $\tau$. This is trivial in the first iteration. So, let $i \geq 2$ and assume that the above statement holds for all previous iterations. Now let $W_i$ be the set of unvisited vertices of $Z$, $B_i$ be the set of all unvisited vertices of $Z'$ and $w_i$ as well as $b_i$ the left most unvisited vertices regarding $\sigma $ and $\tau$, respectively. More formally, it holds that
    \begin{align*}
        W_i &\coloneqq \{v \in V(Z) \mid \pi(v) \geq i\};\\
        B_i &\coloneqq \{v \in V(Z') \mid \pi(v) \geq i\};\\
        w_i &\coloneqq \operatorname{argmin}\{\sigma(v) \mid v \in W_i\};\\
        b_i &\coloneqq \operatorname{argmin} \{ \tau(v) \mid v \in B_i\}.
    \end{align*}
    Furthermore, for every vertex $v \in V(G)$, let $L(v)$ be the label of $v$ in the $i$-th iteration.
    \begin{claim}
        \label{claim:self-similar-1}
        If $W_i$ is nonempty, then there is no vertex $x \in W_i$ with $L(w_i) \prec_\A L(x)$.
    \end{claim}
    \begin{claimproof}
        First, look at the case where $\start$ is in $W_i$. Since $\start$ is the leftmost vertex of $W_i$ in $\sigma$, it is easy to see that in this case $w_i = \start$ and the labels of all vertices $x \in W_i \setminus \{\start\}$ are empty. Thus, using Assumption (\ref{assumption:self-similarity-a}) of this theorem, the claim is trivially true. So, in the following, we can assume that $\start$ is not an element of $W_i$. We argue indirect. Imagine that there would be an $x \in W_i$ with $L(w_i) \prec_\A L(x)$. We know that in all previous iterations the vertices of $Z$ were chosen in the same order as in $\sigma$. So, there is a non-decreasing function $\operatorname f\colon \N^+ \rightarrow \N^+$ such that $\pi(v) = \operatorname f(\sigma(v))$ for all $v \in V(Z)$ with $x \ssigmakleiner w_i$. This fact as well as the fact that $\start \not \in W_i$ and the special structure of $G$ imply the following:
        \begin{alignat*}{3}
            L(x) &= \;\{\operatorname f(\sigma(v)) \mid v \in  N_Z(x), \;\sigma(v) < \sigma(w_i)\} &&= \operatorname f[\preNeighbors \sigma Z x {w_i}]\\
            L(w_i) &= \;\{\operatorname f(\sigma(v)) \mid v \in N_Z(w_i),\; \sigma(v) < \sigma(w_i)\} \;\;&&= \operatorname f[\preNeighbors \sigma Z {w_i} {w_i}]
        \end{alignat*}
        Now applying Assumption (\ref{assumption:self-similarity-b}) of this theorem, we obtain that $L(w_i) \prec_\A L(x)$ implies $\preNeighbors \sigma Z {w_i} {w_i} \prec_\A \preNeighbors \sigma Z x {w_i}$. Thus, $\sigma$ is not an $\A$-ordering of $Z$; a contradiction.
    \end{claimproof}

    Using the same arguments, one can also prove the following. 
    
    \begin{claim}
        \label{claim:self-similar-2}
        If $B_i$ is nonempty, then there is no vertex $x \in B_i$ with $L(b_i) \prec_\A L(x)$.
    \end{claim}

    Finally, we can show that $\pi$ follows $\sigma$ and $\tau$.
    
    \begin{claim}
        \label{claim:slef-similar-3}
        Either $\invpi{i} = b_i$ or $\invpi{i} = w_i$.
    \end{claim}

    \begin{claimproof}
    If either $W_i$ or $B_i$ is empty, then \Cref{claim:slef-similar-3} is a trivial conclusion from \Cref{claim:self-similar-1}, \Cref{claim:self-similar-2} and the definitions of $A^+$ and $\tau\sigma$. So, we may assume $W_i$ and $B_i$ to be nonempty. 

    Suppose $\invpi i$ is not $ b_i$. Then there has been an $x \in V(G)$ with $L(b_i) \prec_\A L(x)$. 
    Now, let $b'$ be any element of $B_i$. If $L(b') \prec_\A L(b_i)$, then we can conclude that $L(b') \prec_\A L(x)$ by transitivity of $\prec_\A$. Otherwise, \Cref{claim:self-similar-2} implies that $L(b')$ and $L(b_i)$ are incomparable and we can use Assumption (b) of this theorem to still conclude that $L(b') \prec_\A L(x)$. In any case, the label of $x$ dominates all labels of vertices in $B_i$. Thus, $\invpi{i}$ has to be an element of $W_i$.

   If $\invpi{i} \neq w_i$, then there is some $y \in V(G)$ such that $L(y)$ dominates $L(w_i)$. The vertex $y$ must be an element of $B_i$ since \Cref{claim:self-similar-1} forbids $L(w_i) \prec_\A L(y)$ for any $y \in W_i$. However, this implies that $L(w_i) \prec_\A L(y) \prec_\A L(x)$. Since $x \in W_i$, this contradicts \Cref{claim:self-similar-1}. Thus, it holds $w_i = \invpi{i}$.
    \end{claimproof}

    \proofsubparagraph{Proof of Property (ii) of \Cref{def:self-similar}} Let $\pi$ be an arbitrary $\A$-ordering of $G$ starting in $Z'$. We want to prove that $\sigma \coloneqq \pi_{\vert V(Z)}$ is an $\A$-ordering of $Z$ starting in $\start$. Due to Assumption (a) of this theorem, we do not visit any vertex without a visited neighbor as long as there are still vertices with visited neighbors. This implies that if $\A$ starts in $Z'$, then it has to pass through $\start$ before it can visit any other vertices of $Z$. Thus, $\sigma$ starts in $\start$.
    
    Now, we want to use the condition formulated in Formula $(\ref{eq:criterium_label_search})$ to prove that $\sigma$ is an $\A$-ordering of $Z$. Let $u,v$ be vertices of $Z$ with $u \ssigmakleiner v$. If $\start$ is an element of $\{u, v\}$, then we use the following argument: We know that $\start$ is the first vertex of $\sigma$ and thus $u = \start$. This already implies $\preNeighbors \sigma Z u u = \preNeighbors \sigma Z v u = \emptyset$. Thus, $\preNeighbors \sigma Z u u $ and $\preNeighbors \sigma Z v u$ are incomparable with respect to $\prec_\A$.

    From now on, we can assume that $\start$ is not an element of $\{u,v\}$. This implies that all neighbors of $u$ and $v$ in $G$ are vertices of $Z$. Since the vertices of $Z$ are ordered in $\sigma$ as they are in $\pi$, we can find a non-decreasing function $\operatorname f\colon \N^+ \rightarrow \N^+$ such that $\pi(x) = \operatorname f(\sigma(x))$ for all $x \in V(Z)$. Therefore, it holds for $w \in \{u,v\}$ that
    \begin{align*}
        \preNeighbors \pi G w u &=  \{\pi(x) \mid x\in N_G(w), \; x \spikleiner u\} \\&= \{\operatorname f(\sigma(x)) \mid x \in N_Z(w), \; x \ssigmakleiner u\} = \operatorname f[\preNeighbors \sigma Z w u].
    \end{align*} 
    We know that $\pi$ is an $\A$-ordering of $G$. This implies that $\preNeighbors \pi G u u \not \prec_\A \preNeighbors \pi G v u$. With the above and assumption (b) of this theorem, we conclude that $\preNeighbors \sigma Z u u \not \prec_\A \preNeighbors \sigma Z v u$, i.e., $\sigma$ is an $\A$-ordering of $Z$.
\end{proof}

Note that not all assumptions in the above theorem seem to be necessary for self-similarity since $\prec_\text{MNS}$ violates Assumption (c) but is still self-similar.

Self-similar searches allow for the following result.

\begin{lemma}
    \label{lem:aequivalenz_rooted_unrooted_instanz}
    Let:
    \begin{enumerate}[-]
        \item $\A$ be a graph search,
        \item $W$ be a graph,
        \item $T_W$ be a spanning tree of $W$,
        \item $\notStart$ be vertex of $W$ such that $T_W \in \G(G,\tau)$ for some $\tau \in \A(W,\notStart)$,
        \item $\start$ be a vertex of $W$ such that $T_W \notin \G(G,\rho)$ for any $\rho \in \A(W,\start)$,
        \item $G$ be graph such that $V(G) \cap V(W) = \start$,
        \item and $T_G$ be a spanning tree of $G$.
    \end{enumerate}
    If $\A$ is self-similar on $G \cup W$, then the following statements are equivalent:
    \begin{enumerate}[(i)]
        \item $T_G$ is a $\G$-tree of an $\A$-ordering $\sigma$ of $G$ starting in $\start$.
        \item $T_G \cup T_W$ is a $\G$-tree of an $\A$-ordering of $G \cup W$.
    \end{enumerate}
\end{lemma}

\begin{proof}
    $(i) \Rightarrow (ii):$ It is easy to see that $(G,W,\start)$ is a cut vertex decomposition of $G \cup W$. $\A$ is self-similar on $G \cup W$. Thus, there is an $\A$-ordering $\pi$ of $G \cup W$ such that the order of the vertices of $G$ in $\pi$ is the same as in $\sigma$ and the order of the vertices of $W$ is the same as in $\tau$. We claim that $T_G \cup T_W$ is a $\G$-tree of $G \cup W$ regarding $\pi$.
    
    Since the vertices of $G$ are ordered in $\pi$ as in $\sigma$, the leftmost vertex of $G$ in $\pi$ must be $\start$. By the same argument we can conclude that the first vertex of $W$ in $\pi$ must be $\notStart$. Especially, the leftmost vertex in $\pi$ is either $\start$ or $\notStart$. Since $\start$ is a vertex of $W$, we know that $\notStart \staukleiner \start$ and therefore the leftmost vertex of $\pi$ must be $\notStart$. It remains to show that every vertex $w$ with $w\neq \notStart$ of $G \cup W$ has exactly one neighbor $u$ in $T_G \cup T_W$ that is to the left of $w$ in $\pi$. 

    If $w$ is a vertex of $W$, then we use the fact that $T_W$ is a $\G$-tree of $W$ regarding $\tau$ which means that there exists exactly one neighbor $u$ of $w$ in $T_W$ with $u \staukleiner w$. By the construction of $\pi$, we know that $u$ is also the only neighbor of $w$ in $T_W$ that is to the left of $w$ in $\pi$. If $w \neq \start$, then all neighbors of $w$ in $T_G \cup T_W$ are also in $T_W$. If $w = \start$, then $w$ might have neighbors in $T_G \cup T_W$ that are not contained in $T_W$. However, we have already seen that $\start$ is the leftmost vertex of $G$ in $\pi$. In both cases we can conclude the uniqueness of $u$.
    The same argumentation works if $w$ is a vertex of $G$.

    $(ii) \Rightarrow (i):$ Let $\pi$ be an $\A$-ordering of $G \cup W$ such that $T_G \cup T_W$ is a $\G$-tree of $G \cup W$ regarding $\pi$. Let $x$ be the leftmost vertex in $\pi$.
    
    First, we consider the case where $x$ is a vertex of $W$. Using property $(ii)$ of \cref{def:self-similar}, we can construct an $\A$-ordering $\sigma$ of $G$ that starts with $\start$ by reducing $\pi$ to vertices of $G$, i.e., $ \sigma \coloneqq \pi_{\vert V(G)}$.
    We claim that $T_G$ is a $\G$-tree of $G$ regarding $\sigma$. Let $y$ be an arbitrary vertex of $G$ with $y \neq \start$.
    $T_G \cup T_W$ is a $\G$-tree of $G \cup W$ regarding $\pi$. Therefore, there exists exactly one neighbor $w$ of $y$ in $T_G \cup T_W$ with $w \spikleiner y$. Since $y$ is a vertex of $G$ and $y \neq \start$ we can conclude that $w$ is a vertex of $G$. The vertices of $G$ are ordered in $\sigma$ as they are in $\pi$. Consequently, $w$ is also the unique neighbor of $y$ in $T_G$ with $y \ssigmakleiner w$. This already implies the desired properties of $T_G$ and $\sigma$.

    Finally, we have to consider the case where $x$ is a vertex of $G$. By swapping the roles of $G$ and $W$ in the above argument we can conclude that $T_W$ is a $\G$-tree of $W$ of an $\A$-ordering starting in $\start$. This contradicts the requirements of this lemma. Therefore, this case cannot occur.
\end{proof}

The above lemma can now be used to give the main result of this paragraph.

\begin{theorem}
    \label{the:Reduktion_rooted_auf_unrooted}
    Let $\A$ be a self-similar search. 
    Let $\C$ be a class of graphs with the following property:
    \[ \forall G,W \in \C\colon \quad \lvert V(G) \cap V(W)\rvert = 1 \quad \Rightarrow \quad G \cup W \in \C\]
    There is a linear-time reduction from the rooted $\G$-tree recognition problem of $\A$ on class $\C$ to the $\G$-tree recognition problem of $\A$ on $\C$.
\end{theorem}
    
\begin{proof}
    Let $(G, T_G, x)$ be an instance the rooted $\G$-tree recognition problem of $\A$ where $G \in \C$. First we assume that there exists a graph $W \in \C$ such that $W$ has a spanning tree $T_W$ and vertices $\notStart$ and $\start$ with the following properties: 
    
    \begin{itemize}
        \item $T_W$ is a $G$-tree of an $\A$-ordering of $W$ starting in $\notStart$, and
        \item there is no ordering in $\A(W,\start)$ with $\G$-tree $T_W$.
    \end{itemize}
    W.l.o.g.~we may assume that $\start = x$ and $V(G) \cap V(W) = \{\start\}$. \Cref{lem:aequivalenz_rooted_unrooted_instanz} implies that there is an $\A$-ordering $\sigma$ of $G \cup W$ such that $T_G \cup T_W \in \G(G \cup W, \sigma)$ if and only if there is an $\A$-ordering $\tau$ of $G$ starting in $\start$ such that $T_G \in \G(G,\tau)$. The condition on $\C$ implies that $G \cup W \in \C$. Furthermore, $G \cup W$ has linear size in the size of $G$ since $W$ is of constant size. Thus, $(G \cup W,T_G \cup T_W)$ is a linear-time reduction from the rooted \G-tree recognition problem of $\A$ on class $\C$ to the \G-tree recognition problem of $\A$ on $\C$.

    It remains to show that a similar reduction can be constructed if there are no $W$, $T_W$, $\notStart$, $\start$ as described above. In that case it holds that 
    $T_G$ is not a $\G$-tree of any $\A$-ordering of $G$ or for every vertex $v$ of $G$ there exists an ordering $\sigma \in \A(G,v)$ such that $T_G \in \G(G,\sigma)$. Either way, in this case our reduction may map an instance $(G, T_G, x)$ of the rooted $\G$-tree recognition problem to the instance $(G, T_G)$ of the unrooted problem.
\end{proof}

Note that this proof is non-constructive. This means that we know that one of the two described reductions exists but we do not know which of the two is the correct one. It might even be possible that it is undecidable whether the graph $W$ with the spanning tree $T_W$ and the vertices $\start$ and $\notStart$ exist in $\C$ or not. Nevertheless, in all the cases that we consider here it is straightforward to find such a graph $W$.

\Cref{the:Reduktion_rooted_auf_unrooted} straightforwardly implies the following result.

\begin{corollary}\label{corol:sss}
    Let $\A$ be a self-similar search. There exists a polynomial-time algorithm for the rooted $\G$-tree recognition problem of $\A$ if and only if there exists a polynomial-time algorithm for the unrooted $\G$-tree recognition problem of $\A$.
\end{corollary}
Note again that the existence of the algorithm does not necessarily imply that we can also specify how it looks like.

\section{NP-Hardness}

As we have seen in \cref{thm:gs}, the (rooted) \G-tree recognition problem of Generic Search is trivial. In this section, we will show that for all other graph searches for which the recognition problem of \cf-trees and \cl-trees have been considered the \G-tree recognition problem is \NP-complete. Note that it trivially holds for these searches that the \G-tree recognition problem is in \NP{} since we can use a search ordering as certificate.

\subsection{BFS and LBFS}

To prove the \NP-hardness of the \G-tree recognition problem of BFS and LBFS, we consider the following problem.

\begin{problem}[Beginning-end-vertex problem of graph search $\A$]~
    \begin{description}
        \item[Instance:] Graph $G$, vertices $s,t \in V(G)$
        \item[Question:] Is there an \A-ordering $\sigma$ that starts with $s$ and ends in $t$?
    \end{description}
\end{problem}

The Beginning-end-vertex problems of BFS and LBFS are \NP-complete as was shown by Charbit et al.~\cite{charbithabibmamcarz14} and Corneil et al.~\cite{corneil2010}, respectively. We will now reduce these problems to the respective \G-tree recognition problem.

\begin{theorem}\label{thm:npc-bfs}
    The (rooted) \G-tree recognition problem of BFS and LBFS is \NP-complete.
\end{theorem}

\begin{proof}
    Let $(G,s,t)$ be an instance of the beginning-end-vertex problem of (L)BFS. Let $L$ be the set of vertices of $G$ whose distance to $s$ is maximal. If $t \notin L$, then $(G,s,t)$ is a trivial no-instance. Therefore, we may assume that $t \in L$. Let $L = \{v_1, \dots, v_k, t\}$. We build a graph $G'$ from $G$ as follows. For every vertex $i \in \{1, \dots, k\}$ we add vertices $u_i$ and $w_i$ to $G$. We make $u_i$ adjacent to $v_i$ and $w_i$ adjacent to $t$. Furthermore, we add the edge $u_iw_i$ to $G'$. 
    
    We create a spanning tree $T$ of $G$ as follows.
    For every vertex $x \in V(G) \setminus \{s\}$, the tree $T$ contains exactly one edge to some vertex $y \in N_G(x)$ with $\dist{G}{s}{y} = \dist{G}{s}{x} - 1$. Furthermore, for every $i \in \{1, \dots, k\}$, the tree $T$ contains the edges $u_iv_i$ and $u_iw_i$.

    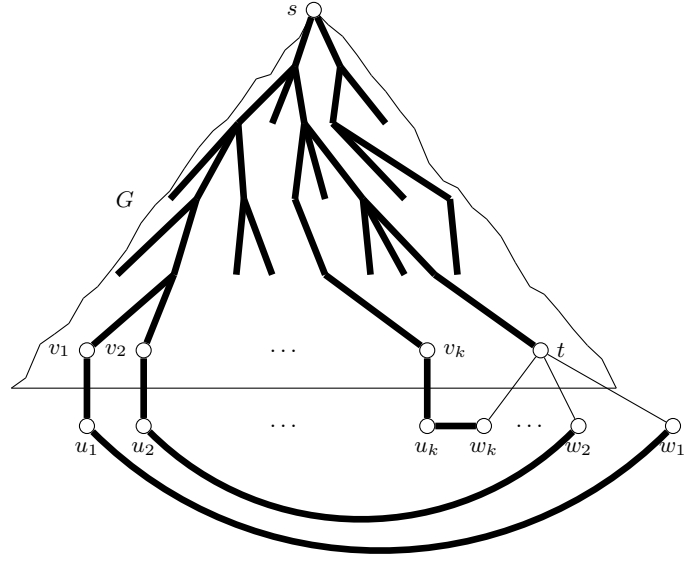
\begin{figure}
    \centering
        \begin{tikzpicture}
        \footnotesize
             \NewDocumentCommand{\irregularline}{%
              O     {2mm}   
              m             
              m            
              D   <> {20}   
            }{{%
              \coordinate (old) at #2;
              \foreach \i in {1,2,...,#4}{
              \draw (old) -- ($ ($#2!\i/(#4+1)!#3$) + (0,#1*rand) $) coordinate (old);
              }
              \draw (old) -- #3;
            }}
            \pgfmathsetseed{41}
            \irregularline[1.5mm]{(0,10)}{(4,5)}
            \pgfmathsetseed{40}
            \irregularline[1.5mm]{(0,10)}{(-4,5)}

            \node[label=180:$s$, vertex] (s) at (0,10) {};

            \draw (4,5) -- (-4,5);

            \coordinate (c1) at (-0.25,9.25);
            \coordinate (c2) at (0.35,9.25);

            \draw[treeedge] (s) -- (c1);
            \draw[treeedge] (s) -- (c2);

            \coordinate (c3) at (-1,8.5);
            \coordinate (c4) at (-0.55,8.5);
            \coordinate (c5) at (-0.125,8.5);
            \coordinate (c6) at (0.25,8.5);
            \coordinate (c7) at (0.95,8.5);
            
            \draw[treeedge] (c1) -- (c3);
            \draw[treeedge] (c1) -- (c4);
            \draw[treeedge] (c1) -- (c5);
            \draw[treeedge] (c2) -- (c6);
            \draw[treeedge] (c2) -- (c7);

            \coordinate (c8) at (-1.9,7.5);
            \coordinate (c9) at (-1.55,7.5);
            \coordinate (c10) at (-0.925,7.5);
            \coordinate (c11) at (-0.25,7.5);
            \coordinate (c12) at (0.15,7.5);
            \coordinate (c13) at (0.65,7.5);
            \coordinate (c14) at (1.2,7.5);
            \coordinate (c15) at (1.8,7.5);
            
            \draw[treeedge] (c3) -- (c8);
            \draw[treeedge] (c3) -- (c9);
            \draw[treeedge] (c3) -- (c10);
            \draw[treeedge] (c5) -- (c11);
            \draw[treeedge] (c5) -- (c12);
            \draw[treeedge] (c5) -- (c13);
            \draw[treeedge] (c6) -- (c14);
            \draw[treeedge] (c6) -- (c15);

            \coordinate (c16) at (-2.6,6.5);
            \coordinate (c17) at (-1.85,6.5);
            \coordinate (c18) at (-1.025,6.5);
            \coordinate (c19) at (-0.55,6.5);
            \coordinate (c20) at (0.15,6.5);
            \coordinate (c21) at (0.75,6.5);
            \coordinate (c22) at (1.2,6.5);
            \coordinate (c23) at (1.6,6.5);
            \coordinate (c24) at (1.9,6.5);
            
            \draw[treeedge] (c9) -- (c16);
            \draw[treeedge] (c9) -- (c17);
            \draw[treeedge] (c10) -- (c18);
            \draw[treeedge] (c10) -- (c19);
            \draw[treeedge] (c11) -- (c20);
            \draw[treeedge] (c13) -- (c21);
            \draw[treeedge] (c13) -- (c22);
            \draw[treeedge] (c13) -- (c23);
            \draw[treeedge] (c15) -- (c24);

            \node (G) at (-2.5,7.5) {$G$};

            \node[vertex, label=180:$v_1$] (v1) at (-3,5.5) {};
            \node[vertex, label=180:$v_2$] (v2) at (-2.25,5.5) {};
            
            \node[] (dots) at (-0.375,5.5) {$\dots$};
            \node[] (dots) at (-0.375,4.5) {$\dots$};
            \node[vertex, label=0:$v_k$] (vk) at (1.5,5.5) {};

            \node[vertex, label=0:$t$] (t) at (3,5.5) {};

            \draw[treeedge] (c17) -- (v1);
            \draw[treeedge] (c17) -- (v2);

            \draw[treeedge] (c20) -- (vk);
            \draw[treeedge] (c23) -- (t);

            \node[vertex, label=270:$u_1$] (u1) at (-3,4.5) {};
            \node[vertex, label=270:$u_2$] (u2) at (-2.25,4.5) {};
            \node[vertex, label=270:$u_k$] (uk) at (1.5,4.5) {};
            \node[vertex, label=270:$w_k$] (wk) at (2.25,4.5) {};
            \node[vertex, label=270:$w_2$] (w2) at (3.5,4.5) {};
            \node[vertex, label=270:$w_1$] (w1) at (4.75,4.5) {};

            \node[] (dots) at (2.875,4.5) {$\dots$};

            \draw[treeedge] (v1) -- (u1);
            \draw[treeedge] (v2) -- (u2);
            \draw[treeedge] (vk) -- (uk);
            \draw (t) -- (w1);
            \draw (t) -- (w2);
            \draw (t) -- (wk);

            \draw[bend angle=45, bend right, treeedge] (u1) to (w1);
            \draw[bend angle=45, bend right, treeedge] (u2) to (w2);
            \draw[treeedge] (uk) to (wk);
        \end{tikzpicture}
        \caption{Construction of the proof of \cref{thm:npc-bfs}. Thick edges represent tree edges. }
    \end{figure}

    We claim that $T$ is $\G$-tree of some (L)BFS of $G'$ starting with $s$ if and only if $t$ is an end vertex of an (L)BFS ordering of $G$ starting with $s$. First assume that $t$ is the end-vertex of the (L)BFS ordering $\sigma$ of $G$ that starts with $s$. W.l.o.g.~we may assume that $\sigma$ ends with the ordering $(v_1, \dots, v_k, t)$. Let $\sigma'$ be the ordering that is constructed from appending $(u_1, \dots, u_k, w_1, \dots, w_k)$ to $\sigma$. It is easy to see that $\sigma'$ is an (L)BFS ordering of $G'$, due to the assumption on the ending of $\sigma$. We claim that $T \in \G(G',\sigma')$. For every vertex $v \in G$ its parent in $T$ has a smaller distance to $s$ than $v$. Thus, the parent is to the left of $v$ in $\sigma$ and, hence, also in $\sigma'$. Furthermore, for every $i$, the parent of $u_i$, namely $v_i$ is to the left of $u_i$ in $\sigma'$ and the parent of $w_i$, namely $u_i$, is to the left of $w_i$ in $\sigma'$. Hence, $T \in \G(G,\sigma')$.

    Now assume that $T \in \G(G', \tau)$ for some (L)BFS ordering $\tau$ of $G$ starting in $s$. For every $i \in \{1,\dots,k\}$, it holds that $u_i \prec_\tau w_i$, due to \cref{lemma:char-g-trees}. By construction of $G'$, it must also hold that $v_i \prec_\tau t$. Therefore, the restriction of $\tau$ to the vertices of $G$ is an (L)BFS ordering of $G$ starting in $s$ and ending in~$t$.

    The hardness of the unrooted problem follows from \cref{the:Reduktion_rooted_auf_unrooted}.
\end{proof}

\subsection{DFS}\label{sec:dfs}

Instead of proving \NP-hardness for the \G-tree recognition problem of DFS, we give a more general result. To this end, we consider the following class of graph searches.

\begin{definition}
    A graph search $\A$ is a \emph{greedy Hamilton search} if it fulfills the following properties:
    \begin{enumerate}[(a)]
        \item For every graph $G$ and every vertex ordering $\sigma$ inducing a Hamiltonian path in $G$ it holds that $\sigma$ is an $\A$-ordering of $G$.
        \item For every graph $G$, every $\sigma \in \A(G)$ and every $v \in V(G)$ it holds that if there is a neighbor $u$ of $v$ in $G$ with $v \ssigmakleiner u$, then the successor of $v$ in $\sigma$ is neighbor of $v$ in $G$.
    \end{enumerate}
\end{definition}

It is clear that DFS is a greedy Hamilton search.

In the following we use a construction that maps every graph $G = (\{v_1,\dots,v_n\},E)$ to a bipartite graph $\bip G$ of similar structure. The construction works by creating new adjacent vertices $\ell_i$ and $u_i$ for every original vertex $v_i$ and making $\ell_i$ and $u_j$ adjacent if and only if $v_i$ and $v_j$ are neighbors in $G$. The same construction was discussed by Brandstädt~\cite{BRANDSTADT1991bip}. He gave the following result.

\begin{lemma}[Brandstädt {\cite[Theorem~2.3~(b)]{BRANDSTADT1991bip}}]
\label{lem:bipG_chordal_bipartite_iff_G_strongly_chordal}
    Let $G$ be graph. The graph $\bip G$ is chordal bipartite if and only if $G$ is strongly chordal.
\end{lemma}

Furthermore, Müller~\cite{MULLER1996291} showed that the Hamiltonian path problem is $\NP$-complete even if the input is restricted to strongly chordal split graphs. With this in mind, we can reduce the Hamilton path problem to solving the $\G-tree$ recognition problem of any greedy Hamilton search.

\begin{figure}
    \centering
    \begin{tikzpicture}[x = 2 cm, y = 1.75 cm,node distance = {1.75 cm}, thick,main/.style = {draw, circle,minimum size = 3.3 em}, every node/.style={transform shape},
    scale=0.8]

    \coordinate[vertex,label = above left:$r$]  (r) at (0,0);
    \coordinate (xdist)  at (1,0);
    \coordinate (xdist2) at (1,0);
    \coordinate (xdist3) at (.5,0);
    \coordinate (ydist1) at (0,.75);
    \coordinate (ydist2) at (0,1);
    \coordinate (ydist3) at (0,0.75);
    \coordinate (ydist4) at (0,0.5);
    \coordinate (yrand)  at (0,0.3);
    \coordinate (xrand)  at (0.3,0);

    \draw [rounded corners,fill=gray!20] 
        ($ (ydist1) - (xdist) - (xrand) - (yrand)$)  -- 
        ($ (ydist1) + (xdist) + (xrand) - (yrand)$) -- 
        ($ (ydist1) + (xdist) + (xrand) + (yrand)$) -- 
        ($ (ydist1) - (xdist) - (xrand) + (yrand)$)  -- 
        cycle;
    \draw [rounded corners,fill=gray!20] 
        ($ (ydist2) + (ydist1) - (xdist) - (xrand) - (yrand)$)  -- 
        ($ (ydist2) + (ydist1) + (xdist) + (xrand) - (yrand)$) -- 
        ($ (ydist2) + (ydist1) + (xdist) + (xrand) + (yrand)$) -- 
        ($ (ydist2) + (ydist1) - (xdist) - (xrand) + (yrand)$)  -- 
        cycle;

    \foreach \i/\l in {1/1,2/2,3/n}
        \coordinate[vertex,label = {[label distance=-0.1cm]181:$\ell_{\l}$}]  (v\l) at
        ($(r) +(ydist1) + \i *(xdist) - 2*(xdist) $);
    \foreach \i in {1,2,n}
        \coordinate[vertex,label = {[label distance=-0.1cm]90:$u_{\i}$}] (v\i') at
        ($ (v\i) + (ydist2)$);
    \coordinate[vertex,label = below left:$\ell_{n+1}$] (vnpp) at 
        ($(vn) + (xdist2) $);
    \coordinate[vertex,label = above left:$u_{n+1}$] (vnpp') at
        ($(vn') + (xdist2) $);
    \coordinate[vertex,label = above left:$m$] (m) at 
        ($(v2') + (ydist3) $);
    

    \draw[rounded corners = 4mm] (r) -- ++($-1.5*(xdist)$) -- ++($ (ydist1) + (ydist2) + (ydist3)$) -- (m);
    \draw[rounded corners = 4mm,treeedge] (r) -- ($(vnpp) - (ydist1) $) -- (vnpp) -- (vnpp') -- ++(ydist3) -- (m);
    
    \draw[treeedge] (r) -- ++($(ydist1) - (yrand)$);
    \draw[] (m) -- ++($(yrand)-(ydist3) $);
    
    \foreach \v in {v1, v2, vn}
        \draw[treeedge] (\v) -- (\v');
    \foreach \i/\k in {1/2, 2/n}
    {
        \draw[] (v\i) -- (v\k');
        \draw[] (v\i') -- (v\k);
    }
    
    \draw[rounded corners] ($(vn') + (xrand)$) -- ($(vn') +0.625*(xdist2)$) -- ++($-1*(ydist2)$) -- (vnpp);
    
    \node at ($0.5*(vn) + 0.5*(v2)$) {...};
    \node at ($0.5*(vn) + 0.5*(v2) + (ydist2)$) {...};

    \end{tikzpicture}
    \caption{The $\NP$-completeness construction for the \G-tree recognition problem of DFS on chordal bipartit graphs. The depicted example is from the case $G = P_3$. The connection of a vertex with a box means that the vertex is connected to all vertices in this box. Tree edges are depicted by thick edges.}\label{fig:cb}
\end{figure}
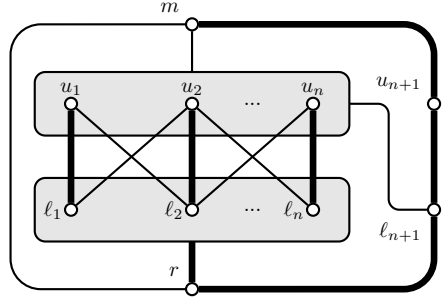

\begin{theorem}
    \label{theorem:DFS-like-hard-on-chordal-biparitit}
    Let $\A$ be a greedy Hamilton search. Then the rooted \G-tree recognition problem of $\A$ is \NP-complete even if the input is restricted to chordal bipartite graphs. The same holds for the unrooted problem if $\A$ is self-similar. 
\end{theorem}
\begin{proof}
    Let $G = (\{v_1,\dots,v_n\},E)$ be an arbitrary strongly chordal graph. We construct $\Gred$ and $\Tred$ using the following steps (see \cref{fig:cb} for an illustration). 
    Construct $\bip G = (\{\ell_1,\dots,\ell_n,u_1,\dots,u_n\},E_{bib})$ as mentioned above. 
    Add a new vertex $\ell_{n
    +1}$ with edges to all $u_i$ for $i\in \{1,\dots,n\}$. 
    Add another new vertex $u_{n+1}$ with an edge to $\ell_{n+1}$.
    Furthermore, add a vertex $r$ and edges to all $\ell_i$ for $i \in \{1,\dots,n+1\}$. Finally, add a vertex $m$ which is adjacent to all $u_i$ for $i \in \{1,\dots,n+1\}$ and $r$. 
    We claim that the resulting graph is chordal bipartite. By using \Cref{lem:bipG_chordal_bipartite_iff_G_strongly_chordal}, we get that $\bip G$ is chordal bipartite. We constructed $\Gred$ from $\bip G$ by successively adding biuniversal vertices and leaves. It is easy to see that these operations maintain the status of being chordal bipartite. 
    
    For $\Tred$, we choose the edges $r\ell_i$ and $\ell_iu_i$ for every $i\in \{1,\dots,n+1\}$ as well as $u_{n+1}m$. Obviously, this construction can be done in polynomial time. It remains to show that there is a Hamiltonian path in $G$ if and only if $\Tred$ is a \G-tree of an $\A$-ordering $\sigma$ of $\Gred$ with $\sigma(r) = 1$.

Suppose, there exists a Hamiltonian path in $G$. W.l.o.g.~we may assume that it has the form $(v_1,\dots,v_n)$. By the construction of $\Gred$, it follows that $\sigma = (r,\ell_1,u_1,\ell_2,u_2,\dots,\ell_{n+1},\allowbreak u_{n+1},m)$ is a Hamiltonian path and, therefore, an $\A$-ordering of $\Gred$. It is easy to check that $\Tred$ is indeed a \G-tree of $\sigma$.

Suppose $\Tred$ is a \G-tree of an $\A$-ordering $\sigma$ of $\Gred$ with $\sigma(r) = 1$. Applying characterization $(iii)$ of \Cref{lemma:char-g-trees} to $\Gred$, $\Tred$, and $\sigma$ yields the statement that $\ell_i \ssigmakleiner u_i$ for all $i \in \{1,\dots,n\}$.
In the following, we will show that it even holds that $\sigma(u_i) = \sigma(\ell_i)+1$. Suppose for contradiction that there is a $j \in \{1,\dots,n\}$ such that $\ell_j$ is leftmost in $\sigma$ of all vertices in $\{\ell_1,\dots,\ell_n\}$ that violate the condition $\sigma(u_i) = \sigma(\ell_i)+1$. Due to the fact that $u_j$ is a neighbor of $\ell_j$ with $\ell_j \ssigmakleiner u_j$, we can use that $\A$ is a greedy Hamilton search. We conclude that the successor of $\ell_j$ in $\sigma$ must be a neighbor of $\ell_j$ in $\Gred$. Due to the construction of $\Gred$ and the fact that $\sigma(r) = 1$, it holds that the successor of $\ell_j$ in $\sigma$ must be $u_k$ for some $k \in \{1,\dots,n\} \setminus \{j\}$. As observed above, $\ell_k$ is to the left of $u_k$ in $\sigma$. Thus, it holds that $\ell_k \ssigmakleiner \ell_j$. This contradicts the choice of $\ell_j$.

Next, we show that $m$ is to the right of $u_i$ for every $i \in \{1,\dots,n+1\}$. Assume for contradiction that $m \prec_\sigma u_j$ for some $j \in \{1,\dots,n+1\}$. We know that $u_j$ is a neighbor of $m$ in $\Gred$. As $\A$ is a greedy Hamilton search, the successor of $m$ in $\sigma$ is a $\Gred$-neighbor of $m$. The neighborhood of $m$ in $\Gred$ consists of the vertices $u_1,\dots,u_{n+1}$ and $r$. Since $\sigma(r) = 1$, the successor of $m$ in $\sigma$ is one of $u_i$. However, this contradicts the fact that $\sigma (u_i) = \sigma (\ell_i) +1$ holds for all $i$.

Summarizing, it holds that $\sigma = (r,\ell_{\pi(1)},u_{\pi(1)},\dots,\ell_{\pi(n+1)},u_{\pi(n+1)},m)$
where $\pi$ is a permutation of $\{1,\dots,n\}$. For every $i \in \{1,\dots,n-1\}$, vertex $m$ is a neighbor of $u_{\pi(i)}$ and $u_{\pi(i)} \ssigmakleiner m$. Since $\A$ is a greedy Hamilton search, we conclude that $u_i$ and $\ell_{i+1}$ are neighbors for every $i \in \{1,\dots,n-1\}$. The construction of $\Gred$ yields that $(v_{\pi(1)},\dots,v_{\pi(n)})$ is a Hamiltonian path in $G$. 

Note that the union of two chordal bipartite graphs having only one vertex in common is chordal bipartite. Using \cref{the:Reduktion_rooted_auf_unrooted}, we can conclude the hardness of the unrooted case if $\A$ is self-similar.
\end{proof}

The above proof can be easily adapted to show hardness in the case where $G$ is a split graph.

\begin{theorem}
\label{theorem:DFS-like-hard-on-split}
    Let $\A$ be a greedy Hamilton search. Then the rooted \G-tree recognition problem of $\A$ is \NP-complete even if the input is restricted to split graphs.
\end{theorem}

\begin{proof}
    We construct $\Gred$ and $\Tred$ as described in the proof of \Cref{theorem:DFS-like-hard-on-chordal-biparitit} with the sole difference that we add additional edges to make $\{r,u_1,\dots,u_{n+1},m\}$ a clique in $\Gred$. Note that $l_1,\dots,l_{n+1}$ form an independent set of $\Gred$ and thus $\Gred$ is a split graph.

    The arguments we provided in the proof of \Cref{theorem:DFS-like-hard-on-chordal-biparitit} can again be used to show that $\Tred$ is a $\G$-tree rooted in $r$ of $\Gred$ if and only if $G$ includes a Hamilton path.
\end{proof}

Note that split graphs do not fulfill the property given in \cref{the:Reduktion_rooted_auf_unrooted}. Hence, we can not simply extend the result to the unrooted case. Nevertheless, the superclass of chordal graphs fulfills the property which implies the following.

\begin{corollary}
     Let $\A$ be a greedy Hamilton search. Then the rooted \G-tree recognition problem of $\A$ is \NP-complete even if the input is restricted to chordal graphs. The same holds for the unrooted problem if $\A$ is self-similar.
\end{corollary}

\subsection{LDFS, MCS, and MNS}

We adapt a proof given in \cite{beisegel2021recognition} for the \NP-completeness of the \F-tree recognition problem of that searches. We use a polynomial-time reduction from 3-SAT. Let $\cal{I}$ be an instance of 3-SAT with the set of literals $x_1, \dots, x_k,\overline{x}_1,\ldots,\overline{x}_k$ and the clauses $c_1, \ldots ,c_\ell$. We construct the corresponding graph $G(\mathcal{I})$ as follows (see \cref{fig:mns-tree} for an example). Every literal and every clause is represented by a vertex with the same name. The literal vertices induce the complement of the matching in which $x_i$ is matched to $ \overline{x}_i $ for every $ i \in \{1,\ldots, k\} $. The clause vertices are pairwise non-adjacent and every vertex $ c_i $ is adjacent to each literal vertex, except those representing the literals of the clause $ c_i $ for every $ i \in \{1, \ldots, \ell\} $. Additionally, we add the vertices $r$, $p$, $q$, and $b$. The vertices $r$, $p$, and $q$ are adjacent to all literal vertices and all clause vertices while $b$ is adjacent to all literal vertices. Finally, we add the edges $pr$, $rb$, $rq$, and $bq$. The spanning tree $T(\I)$ of $G(\I)$ consists of all edges incident to $p$ and the edges $rb$ and $bq$.

\begin{figure}
		\centering
    \begin{tikzpicture}[scale=0.8,vertex/.style={inner sep=2pt,draw,circle,fill=white},
	noedge/.style={dashed},
    every node/.style={transform shape}
	]
	\large
	\draw[rounded corners=5pt]
 	(-2.5,4) rectangle (6.5,9);
 	\draw[rounded corners=5pt, fill=black!10!white]
 	(-2,7) rectangle (7,8.75);

	\node[vertex,label={180:$x_1$}] (x1) at (-1, 8.4) {};
	\node[vertex,label={180:$\overline{x}_1$}] (nx1) at (-1, 7.4) {};
	\node[vertex,label={0:$x_2$}] (x2) at (1, 8.4) {};
	\node[vertex,label={0:$\overline{x}_2$}] (nx2) at (1, 7.4) {};
	\node[vertex,label={0:$x_3$}] (x3) at (3, 8.4) {};
	\node[vertex,label={0:$\overline{x}_3$}] (nx3) at (3, 7.4) {};
	\node[vertex,label={0:$x_4$}] (x4) at (5, 8.4) {};
	\node[vertex,label={0:$\overline{x}_4$}] (nx4) at (5, 7.4) {};
	
	\node[vertex,label={[label distance=0.5cm]270:$\overline{x}_1 \lor x_2 \lor \overline{x}_3$}] (c1) at (-1, 5.25) {};
	\node[vertex,label={[label distance=0.5cm]270:${x}_1 \lor \overline{x}_3 \lor {x}_4$}] (c2) at (2,5.25) {};
	\node[vertex,label={[label distance=0.5cm]270:$\overline{x}_1 \lor \overline{x}_3 \lor \overline{x}_4$}] (c3) at (5, 5.25) {};

	\node[vertex,label={0:$b$}] (b) at (9,8) {};
    \node[vertex,label={0:$p$}] (p) at (7.5,6.5) {};
	\node[vertex,label={0:$r$}] (r) at (7.5,5.25) {};
	\node[vertex,label={-90:$q$}] (q) at (9,4.5) {};
	
	\draw[treeedge] (b) -- (q);
	\draw[treeedge] (r) to (b);
    \draw[treeedge] (r) to (p);
    \draw (r) -- (q);
 
	\draw[noedge] (x1) -- (nx1);
	\draw[noedge] (x2) -- (nx2);
	\draw[noedge] (x3) -- (nx3);
	\draw[noedge] (x4) -- (nx4);
  
    \draw[noedge] (c1) -- (c2);
    \draw[noedge] (c2) -- (c3);
    \draw[noedge, bend angle=20, bend right] (c1) to (c3);
	
	\draw[noedge] (c1) -- (nx1);
	\draw[noedge] (c1) -- (x2);
	\draw[noedge] (c1) -- (nx3);
	\draw[noedge] (c2) -- (x1);
	\draw[noedge] (c2) -- (nx3);
	\draw[noedge] (c2) -- (x4);
	\draw[noedge] (c3) -- (nx1);
	\draw[noedge] (c3) to (nx3);
	\draw[noedge] (c3) -- (nx4);
  
  \draw[treeedge] (p) -- (6.5,6.6)--(p)--(6.5,6.4)--(p)--(6.5,6.8)--(p)--(6.5,6.2);
  
  \draw[] (b) -- (7,8.1)--(b)--(7,7.9)--(b)--(7,7.7)--(b)--(7,8.3);
	
  \draw (r) -- (6.5,5.35) --(r)--(6.5,5.15)--(r)--(6.5,5.55)--(r)--(6.5,4.95);
  
	\draw[] (q) -- (6.5,4.6) --(q)--(6.5,4.8)--(q)--(6.5,4.4)--(q)--(6.5,4.2);
	
	\end{tikzpicture}
 	\caption{The $\NP$-completeness construction for the \G-tree recognition problem of MNS. The depicted graph is $G(\mathcal{I})$ for $\mathcal{I} = (\overline{x}_1 \lor x_2 \lor \overline{x}_3) \land ({x}_1 \lor \overline{x}_3 \lor {x}_4) \land (\overline{x}_1 \lor \overline{x}_3 \lor \overline{x}_4)$. In both boxes only non-edges are displayed by dashed lines. The connection of a vertex with a box means that the vertex is connected to all vertices in this box. Tree edges are depicted by thick edges. Note that this figure is an adaption of a figure given in \cite{beisegel2021recognition}.}\label{fig:mns-tree}
\end{figure}
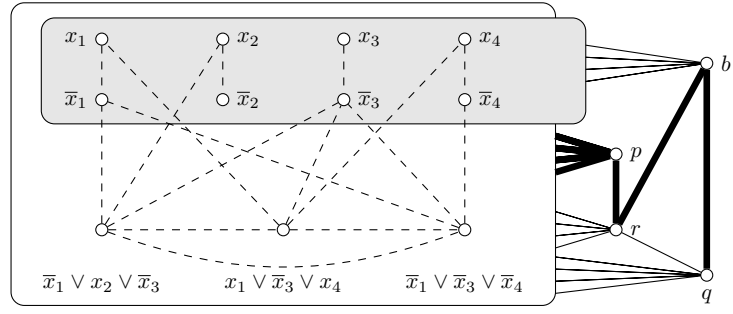

\begin{lemma}\label{lemma:b-before-clause}
    If $T(\I)$ is a \G-tree of an MNS ordering $\sigma$ of $G(\I)$ with $\sigma(r) = 1$, then vertex $b$ is to the left of every clause vertex in $\sigma$.
\end{lemma}%
\begin{proof}
    As $b$ is the parent of $q$ in $T(\I)$ if $T(\I)$ is rooted in $r$, it follows that $b$ has to be to the left of $q$ in $\sigma$. It holds that $N_{G(\I)}(b) \setminus \{q\} \subsetneq N_{G(\I)}(q)$ and all clause vertices are elements of $N_{G(\I)}(q) \setminus N_{G(\I)}(b)$. Thus, if some clause vertex is to the left of $b$ in $\sigma$, vertex $q$ is also to the left of $b$ in $\sigma$; a contradiction.
\end{proof}
\begin{lemma}\label{lemma:mns-start}
    If $T(\I)$ is a \G-tree of an MNS ordering $\sigma$ of $G(\I)$ with $\sigma(r) = 1$, then $\sigma(p) = 2$ and $\invsig 3 , \dots, \invsig {k+2}$ forms an assignment of \I. 
\end{lemma}

\begin{proof}
    In any \G-tree of $\sigma$ it holds that $\invsig 2$ is adjacent to $\invsig 1 $. Therefore, $\invsig 2$ has to be $p$ or $b$. If $\invsig 2$ equals $b$, then the next two vertices will be vertices that are adjacent to $b$ and $r$ ($q$ and/or some literal vertices). At least one of these two vertices has $p$ as parent in $T(\I)$ and, thus, has to be to the right of $p$ in $\sigma$; a contradiction. Therefore, the second vertex of $\sigma$ is $p$. Now $\sigma$ has to visit some vertex that is neighbor of both $r$ and $p$. Due to \cref{lemma:b-before-clause}, this vertex has to be a literal vertex. As long as there are variables whose two literal vertices are not visited so far, these literal vertices have the unique maximal label (except from clause vertices which cannot be visited before $b$, due to \cref{lemma:b-before-clause}). Therefore, $\sigma$ visits a whole assignment after $r$ and~$p$.    
\end{proof}

\begin{lemma}\label{lemma:mns-reduction}
    A 3-SAT instance $\I$ has a satisfying assignment if and only if $T(\I)$ is an \G-tree of some MNS (MCS, LDFS) ordering of $G(\I)$ starting with vertex $r$.
\end{lemma}

\begin{proof}
    First we prove that $\I$ has a fulfilling assignment if $T(\I)$ is a $\G$-tree of some MNS ordering $\sigma$ of $G(\I)$ starting with $r$. As we have seen in \cref{lemma:mns-start}, the ordering $\sigma$ visits $p$ after $r$ and then it visits an assignment of $\I$. If this assignment does not fulfill $\I$, then there is a clause vertex $c_i$ which is adjacent to all vertices that were visited so far. On the other hand, all unvisited literal vertices and $b$ and $q$ have at least one visited non-neighbor. Therefore, the next vertex that is visited by $\sigma$ is a clause vertex and, thus, there is a clause vertex that is to the left of $b$ in $\sigma$; a contradiction to \cref{lemma:b-before-clause}.

    Now assume that there is a fulfilling assignment $\B$ of $\I$. We construct an LDFS (MCS) ordering $\sigma$ with $\G$-tree $T(\I)$ as follows. We start in $r$, then we choose $p$ and then we choose the literal vertices that correspond to $\B$. As these vertices form a clique, this is a valid prefix of MCS and LDFS. Since $\B$ is a fulfilling assignment, every unvisited literal vertex and every clause vertex has at least one non-neighbor that was already visited. The vertices $b$ and $q$ have exactly one visited non-neighbor (vertex $p$). Therefore, $b$ has the maximal MCS label and can be visited next by MCS. It is easy to see that $T(\I)$ forms a \G-tree of any ordering starting with these vertices. For LDFS we observe that $b$ and $q$ are adjacent to all visited literal vertices while all other unvisited vertices are not adjacent to at least one of them. Therefore, the LDFS label of $b$ and $q$ is also maximal and $b$ can be chosen next. This completes the proof.  
\end{proof}

Note that the graph $G(\I)$ that we have constructed is an induced subgraph of the graph $G(\I)$ that was constructed for the proof of Theorem 5.1 in \cite{beisegel2021recognition}. As this graph is weakly chordal, our graph is it as well. Combining this fact with \cref{lemma:mns-reduction}, we get the \NP-completeness of the rooted problem. We observe that the union of two weakly chordal graphs that have exactly one vertex in common is also weakly chordal. Therefore, \cref{the:Reduktion_rooted_auf_unrooted} also implies the \NP-completeness of the unrooted problem.

\begin{theorem}
    The (rooted) $\G$-tree recognition problem of LDFS, MCS, and MNS is \NP-complete even if the input is restricted to weakly chordal graphs.
\end{theorem}

\section{Polynomial-time Algorithms}
\subsection{Layered Searches and Bipartite Graphs}

\begin{lemma}
    \label{lemma:G_Baum_impliziert_gleiche_abstaende}
    Let $\A$ be a layered search, $G$ be a bipartite graph, $\start \in V(G)$, $\sigma \in \A(G,\start)$ and $T \in \G(G,\sigma)$. For all vertices $v$ of G it holds that $\dist G \start v  = \dist T \start v$.
\end{lemma}
\begin{proof}
    Let $M \coloneqq \{ v \in V(G) \mid \dist G \start v  \neq \dist T \start v\}$. Suppose $M$ is not empty. Let $x$ be an element of $M$ with minimal distance to $\start$ in $T$. Furthermore, let $w$ be the predecessor of $x$ in the $\start$-$x$-path in $T$. By the choice of $x$, it holds that $\dist{T}{\start}{w} = \dist{G}{\start}{w}$ and, thus, $\dist{T}{\start}{x} = \dist{G}{\start}{w} + 1$.
    As $x$ and $w$ are neighbors in $G$, it holds that 
    $\dist{G}{\start}{x} \leq \dist{G}{\start}{w} +1$.
    Since $x \in M$, it holds $\dist{G}{\start}{x} < \dist{T}{\start}{x}$ and, thus, $\dist{G}{\start}{x} < \dist{G}{\start}{w} +1$. As $G$ is bipartite, it is not possible that $\dist{G}{\start}{x} = \dist{G}{\start}{w}$. As a consequence, we have $\dist{G}{\start}{x} < \dist{G}{\start}{w}$.
    However, $\A$ is a layered search and, thus, the above implies that $x$ is to the left of $w$ in $\sigma$, a contradiction to the fact that $T$ is a $\G$-tree of $\sigma$ and the choice of $w$.
\end{proof}

This lemma implies a characterization of \G-trees of layered searches on bipartite graphs.

\begin{theorem}
    Let $\A$ be a natural layered graph search, $G$ be a bipartite graph, $T$ be a spanning tree of $G$ and $\start \in V(G)$. The following statements are equivalent:
    \begin{enumerate}[(i)]
        \item There exists an $\A$-ordering $\sigma$ of $G$ such that $\sigma$ starts in $\start$ and $T$ is $\G$-tree of $\sigma$.
        \item For all vertices $u\in V$ it holds that the distances of $\start$ and $u$ in $G$  is equal to the distance of $\start$ and $u$ in $T$.
        \item $T$ is a $\G$-tree of every $\A$-ordering of $G$ starting in $\start$.
    \end{enumerate}
    
\end{theorem}

\begin{proof}
    $(i) \Rightarrow (ii)$ follows from \cref{lemma:G_Baum_impliziert_gleiche_abstaende}. So, suppose that $(ii)$ holds. Let $\sigma$ be an arbitrary $\A$-ordering of $G$ starting in $\start$. Let $u,w$ be two distinct vertices of $G$ such that $u$ lies on the $\start$-$w$-path in $T$. It follows that the distance of $\start$ and $u$ in $T$ is smaller than the distance of $\start$ and $w$ in $T$. Using $(ii)$, it follows that the same holds for the distances of these vertices in $G$. Since $\A$ is a layered search, $u$ has to be to the left of $w$ in $\sigma$. This implies $(iii)$.

    As $\A$ is a natural graph search, $\A(G,\start)$ is not empty. Hence, $(iii) \Rightarrow (i)$ holds trivially.
\end{proof}

As the distances from the start vertex can be computed in linear time, the theorem implies the following running time bound.

\begin{corollary}
    The rooted $\G$-tree recognition problem of a natural layered graph search is linear time solvable if the input is restricted to bipartite graphs.
\end{corollary}

\subsection{Complete Bipartite Graphs}

Although the class of complete bipartite graphs seems to be quite restrictive, the Partial Search Order problem of DFS, LDFS, and MCS is \NP-complete on this class~\cite{beisegel2024computing,scheffler2025partial}. We will show here that this does not hold for the \G-tree recognition. To this end, we show that the \G-tree recognition problem is trivial for the following family of graph searches that includes these three searches. 

\begin{definition}
    Let $G$ be a complete bipartite graph with bipartition $(X,Y)$ and let $\sigma$ be a vertex ordering of $G$. Let $i$ be the largest index such that $\invsig i \in X$ and let $j$ be the largest index such that $\invsig j \in Y$. Let $k = \min\{i,j\}$. The ordering $\sigma$ is \emph{alternating} if for all $\ell, \ell' \leq k$ it holds that $\invsig \ell \invsig {\ell'} \in E(G)$ if and only if $ \ell \mod 2 \equiv (\ell' + 1) \mod 2$.

    We say that a graph search $\A$ \emph{can alternate} on $G$ if every alternating vertex ordering of $G$ is an $\A$-ordering of $G$.
\end{definition}

\begin{theorem}
    Let $G$ be a complete bipartite graph and let $\A$ be graph search that can alternate on $G$. Then for each vertex $\start \in V(G)$ and each spanning tree $T$ of $G$ there is an ordering $\sigma \in \A(G,\start)$ such that $T \in \G(G,\sigma)$.
\end{theorem}

\begin{proof}
    We assume in the following that $T$ is rooted in $\start$. We first mark all vertices as \emph{inactive}. Now we visit vertex~$\start$. Whenever a vertex $v$ is visited, all its children in $T$ are marked as \emph{active} while $v$ is marked as inactive. The set $S_v$ contains those neighbors of $v$ in $G$ that are active if there is any unvisited neighbor of $v$. Otherwise, $S_v$ contains all the active non-neighbors of $v$. If there is a vertex in $S_v$ that is not a leaf in $T$, then we visit one of these vertices next. Otherwise, we visit an arbitrary vertex in $S_v$. 
    
    We claim that this algorithm computes an \A-ordering $\sigma$ with $T \in \G(G,\sigma)$. First we show that there is always a vertex in $S_v$ unless all the vertices of $G$ have already been visited. First assume that there is no unvisited neighbor of $v$. Then all vertices of the partition set not containing $v$ are already visited. In that case, for all vertices in the partition set of $v$ it holds that their parent in $T$ has already been visited as those parents are contained in the other partition set. Hence, every unvisited vertex is active and, thus, it is contained in $S_v$.
    
    Therefore, we may assume that there are unvisited neighbors of $v$. Let $u$ be the vertex visited before $v$. Assume for contradiction that there is no active neighbor of $v$. As our algorithm only visits vertices that are active and vertices become only active when their parent is visited, no child of $v$ in $T$ can be visited before $v$. After $v$ is visited, all children of $v$ are active. As we have assumed that there is no active neighbor of $v$, it must hold that $v$ is a leaf in $T$. Let $w$ be some unvisited neighbor of $v$. Consider the path $P$ from $w$ to $\start$ in $T$. Let $y$ be the first vertex on $P$ that has already been visited. The child $x$ of $y$ on the path from $w$ to $\start$ is an active vertex. Therefore, $x$ is not a neighbor of $v$ and, thus, not equal to $w$. However, $x$ has a child in $T$ since it is not equal to $w$. This is a contradiction since $v$ was visited before $x$ by the algorithm but both $v$ and $x$ were active neighbors of $u$ and $v$ is a leaf in $T$ but $x$ is not.

    Since vertices are only visited if their parent has already been visited, $T$ is a \G-tree of $\sigma$. By the choice of $S_v$, $\sigma$ is an alternating ordering and, thus, an \A-ordering of $G$.
\end{proof}

It is easy to see that DFS, LDFS, MCS, and MNS can alternate on every complete bipartite graph.

\begin{corollary}
    Every spanning tree of a complete bipartite graph is a \G-tree of a DFS ordering, an LDFS ordering, an MCS ordering and an MNS ordering.
\end{corollary}

This does not hold for layered searches. Instead, the following characterization holds.

\begin{theorem}
     Let $\A$ be any natural and layered graph search. A spanning tree $T$ of a complete bipartite graph $G$ is a \G-tree of $\A$ rooted in $r \in V(G)$ if and only if $N_G(r) = N_T(r)$.  
\end{theorem}

\begin{proof}
    Let $X$ and $Y$ be the two sets of the bipartition of $G$. First assume that $N_G(r) = N_T(r)$ for some vertex $r \in V(G)$. W.l.o.g. we may assume that $r \in X$. We start the search ordering of $\A$ in $r$. Since $\A$ is layered, $r$ is followed by all vertices of $Y$ in some ordering. As $r$ is adjacent to all of these vertices in $T$, this does not contradict the edges of $T$. Since the remaining vertices of $X$ have their parents in $Y$, this ordering has $T$ as its \G-tree.

    Now assume that $T$ is a \G-tree of an \A-ordering $\sigma$ starting in some vertex $r \in V(G)$. Again, we may assume w.l.o.g. that $r \in X$. As $\A$ is layered, $r$ is followed by the vertices of $Y$ in $\sigma$. Therefore, $r$ is the only neighbor of the vertices of $Y$ that is to the left of them in~$\sigma$. This implies that $r$ has to be the parent of every vertex of $Y$ in $T$.
\end{proof}

\subsection{Chordal Graphs}\label{sec:chordal}

Both the \cf-tree and the \cl-tree recognition problem of a large family of graph searches have been shown to be polynomial-time solvable on chordal graphs~\cite{beisegel2021recognition}. Using known results on the Partial Search Order problem, we can show the same holds for the \G-tree recognition of MCS. Rong et al.~\cite{rong2026partial} give an $\O(n^4)$ time algorithm for the PSOP on chordal graphs with $n$ vertices.\footnote{The concrete running time bound is only given in a preliminary arXiv version of the paper~\cite{rong2023polynomial}.} For split graphs, this can be improved to linear time~\cite{scheffler2025partial}. These results imply the following.
\begin{proposition}\label{prop:mcs-chordal}
     Let $G$ be a graph with $n$~vertices. The rooted \G-tree recognition problem of MCS can be solved for $G$
    \begin{enumerate}
        \item in $\O(n^4)$ time if $G$ is a chordal graph, and
        \item in $\O(n^2)$ time if $G$ is a split graph.
    \end{enumerate}
\end{proposition}

In the remainder of this section we will generalize this result as it was done for \cf-trees and \cl-trees~\cite{beisegel2021recognition}. We first define the family of strictly comp-monotone graph searches.\footnote{Beisegel et al.~\cite{beisegel2021recognition} defined the family of \emph{edge-forced PEO-finders}. This family is equivalent to the strictly comp-monotone searches when chordal graphs are considered.}

\begin{definition}
A vertex ordering $\sigma = (v_1, \ldots, v_n)$ of a graph $G$ has the \emph{comp-inclusion property} if for each pair $i,j \in \{1,\ldots,n\}$ with $i < j$ such that $v_i$ and $v_j$ are part of the same component of $G[v_i, \ldots, v_n]$ it holds that $N_G(v_i) \cap \{v_1,\ldots,v_{i-1}\}$ is not a strict subset of $N_G(v_j) \cap \{v_1,\ldots,v_{i-1}\}$. A graph search $\A$ has the \emph{comp-inclusion property} if every $\A$-ordering of every graph $G$ has the comp-inclusion property.

A graph search $\A$ has the \emph{comp-tie property} if for each graph $G$, each (possibly empty) $\A$-prefix $(v_1, \ldots, v_k)$ of $G$ and each pair of vertices $x,y$ which are in same component of $G - \{v_1, \ldots, v_k\}$ the following holds: If $N_G(x) \cap \{v_1, \ldots, v_k\} = N_G(y) \cap \{v_1, \ldots, v_k\}$, then $(v_1, \ldots, v_k, x)$ is an $\A$-prefix of $G$ if and only if $(v_1, \ldots, v_k, y)$ is an $\A$-prefix of $G$.

A graph search $\A$ is called \emph{strictly comp-monotone} if it has both the comp-inclusion property and the comp-tie property.
\end{definition}

Shier~\cite{shier1984some} defined the search \emph{Maximal Element in Component Scheme (MEC)}. It chooses a vertex next whose set of visited vertices is inclusion maximal among all the vertices in the same component of the graph that is induced by the unvisited vertices. Thus, MEC is the most general graph search having the comp-inclusion property. Shier~\cite{shier1984some} showed that a vertex ordering of a chordal graph is a PEO if and only if it is a MEC ordering. This implies the following theorem.

\begin{theorem}[see~Shier~{\cite[Theorem~2]{shier1984some}}]\label{thm:searches-chordal}
A vertex ordering $\sigma$ of a chordal graph $G$ is a PEO of $G$ if and only if $\sigma$ has the comp-inclusion property.
\end{theorem}

Thus, graph searches having the comp-inclusion property compute PEOs of chordal graphs.

\begin{corollary}\label{corol:search-peo}
For every graph search $\A$ having the comp-inclusion property, every $\A$-ordering of a chordal graph $G$ is a PEO of $G$. 
\end{corollary}

In the following, we will show that for all comp-monotone graph searches, the set of rooted \G-trees of a chordal graph is the same. This has been shown by Beisegel et al.~\cite{beisegel2021recognition} for \cf-trees and \cl-trees. We will need the following lemmas.

\begin{lemma}[Beisegel et al.~{\cite[Lemma 6.3]{beisegel2021recognition}}]\label{lemma:peo_path}
Let $G$ be a chordal graph and let $\sigma$ be a PEO of $G$. Let $P = (v_1, \ldots, v_k)$ be an induced path of $G$. If $v_1$ is the leftmost vertex of $P$ in $\sigma$, then $v_1 \prec_\sigma v_2 \prec_\sigma \ldots \prec_\sigma v_k$.
\end{lemma}

\begin{lemma}[Shier~{\cite[proof of Theorem~2]{shier1984some}}]\label{lemma:chordal_neighborhoods}
Let $G$ be a chordal graph and let $\sigma = (v_1, \ldots, v_n)$ be a PEO of $G$. Let $S = \{v_1, \ldots, v_{i-1}\}$ and let $C$ be the connected component of $G - S$ containing $v_i$. Then for every $w \in C$ it holds that $N_G(w) \cap S \subseteq N_G(v_i)$.
\end{lemma}

Now we are able to prove that strictly comp-monotone graph searches behave very similar on a chordal graph if they are given the same tie-breaker. Let $x$ be an arbitrary vertex of the chordal graph. Consider the component containing vertex $x$ that is induced by the vertices to the right of $x$ in such a search ordering. We show that this component is independent from the particular search and only depends on the tie-breaker.

\begin{lemma}\label{lemma:chordal-plus-comp}
Let $G$ be a chordal graph and $\rho$ be a vertex ordering of $G$. Let $\A_1$ and $\A_2$ be two strictly comp-monotone graph searches and let $\sigma_1 = \A_1^+(G,\rho)$ and $\sigma_2 = \A_2^+(G,\rho)$. For every $v\in V(G)$, let $S_1(v) = \{w \mid w \prec_{\sigma_1} v\}$, let $S_2(v) = \{w \mid w \prec_{\sigma_2} v\}$ and let $C_1(v)$ and $C_2(v)$ be the components of $G[V(G) \setminus S_1(v)]$ and $G[V(G) \setminus S_2(v)]$, respectively, containing $v$. 
Then for every vertex $v \in V(G)$, $C_1(v)$ is equal to $C_2(v)$.
\end{lemma}

\begin{proof}
Assume for contradiction that the statement does not hold. Let $x$ be the leftmost vertex in $\sigma_1$ such that $C_1(x) \neq C_2(x)$. First note that, by \cref{corol:search-peo}, both $\sigma_1$ and $\sigma_2$ are PEOs of $G$. 

\begin{claim}\label{claim:chordal-plus-comp1}
    If for $i,j \in \{1,2\}$ it holds that $C_j(x) \setminus C_i(x) \neq \emptyset$, then it holds for the leftmost vertex $y \in C_j(x) \setminus C_i(x)$ in $\sigma_j$ that $x \in C_i(y)$.
\end{claim}

\begin{claimproof}
    Let $P$ be an induced path from $x$ to $y$ in $C_j(x)$. Since $x$ is the leftmost vertex of $P$ in $\sigma_j$, it follows from \cref{lemma:peo_path} that all inner vertices of $P$ are to the left of $y$ in $\sigma_j$. By the choice of $y$, all these inner vertices are elements of $C_i(x)$ and, thus, they are to the right of $x$ in $\sigma_i$. If $y$ would also be to the right of $x$ in $\sigma_i$, then all vertices of $P$ except $x$ would be to the right of $x$ in $\sigma_i$ and, thus, $P$ would be a path in $C_i(x)$; a contradiction to the fact that $y \notin C_i(x)$. Therefore, $y$ is to the left of $x$ in $\sigma_i$. Hence, all the vertices of $P$ except $y$ are to the right of $y$ in $\sigma_i$ and, thus, $x \in C_i(y)$.
\end{claimproof}

\begin{claim}\label{claim:chordal-plus-comp2}
    $C_1(x) \setminus C_2(x) \neq \emptyset$.
\end{claim}

\begin{claimproof}
    Assume for contradiction that this is not the case. Since $C_1(x) \neq C_2(x)$, we then know that $C_2(x) \setminus C_1(x) \neq \emptyset$. Let $y$ be the leftmost vertex of $C_2(x) \setminus C_1(x)$ in $\sigma_2$. Due to \cref{claim:chordal-plus-comp1}, $x \in C_1(y)$ and, thus, $y \prec_{\sigma_1} x$. By the choice of $x$, it holds that $C_1(y) = C_2(y)$. Thus, $x \in C_2(y)$. Therefore, $x \in C_2(y)$ and $y \in C_2(x)$. This is only possible if $x = y$. However, then $y \in C_1(x) \cap C_2(x)$;  a contradiction to the fact that $y \in C_2(x) \setminus C_1(x)$
\end{claimproof}

Let $u$ be the leftmost vertex of $C_1(x) \setminus C_2(x)$ in $\sigma_2$. 

\begin{claim}\label{claim:chordal-plus-comp3}
    $x \in C_2(u)$.
\end{claim}

\begin{claimproof}
    Let $y$ be the leftmost vertex of $C_1(x) \setminus C_2(x)$ in $\sigma_1$. \Cref{claim:chordal-plus-comp1} implies that $x \in C_2(y)$ and, thus, $y$ is to the left of $x$ in $\sigma_2$. Since $u$ is to the left of $y$ in $\sigma_2$ (or $u$ is equal to $y$), it also holds that $u$ is to the left of $x$ in $\sigma_2$. 
    
    Let $w$ be an arbitrary vertex in $C_1(x) \setminus \{u\}$. We distinguish two cases. First assume that $w \in C_2(x)$. Then $w \in C_2(y)$ since $x \in C_2(y)$. This implies that $y \prec_{\sigma_2} w$ and, hence, $u \prec_{\sigma_2} w$. Now assume that $w \notin C_2(x)$. Then $u \prec_{\sigma_2} w$, due to the choice of~$u$. 
    
    Thus, in any case all the vertices of $C_1(x) \setminus \{u\}$ are to the right of $u$ in $\sigma_2$. Hence, the path between $u$ and $x$ in $C_1(x)$ is also contained in $C_2(u)$ and, therefore, $x \in C_2(u)$.
\end{claimproof}

\begin{claim}\label{claim:chordal-plus-comp4}
    $N_G(x) \cap S_1(x) = N_G(u) \cap S_1(x)$.
\end{claim}

\begin{claimproof}
    Let $z \in N_G(x) \cap S_1(x)$. It holds that $x \in C_1(z)$ and, since $u \in C_1(x)$, it holds that  $u \in C_1(z)$. Since $z \prec_{\sigma_1} x$, the choice of $x$ implies that $C_1(z) = C_2(z)$ and, as $u \in C_1(z)$, it holds $u \in C_2(z)$. Combining this with \cref{claim:chordal-plus-comp3}, we see that $z \prec_{\sigma_2} u \prec_{\sigma_2} x$. By \cref{lemma:chordal_neighborhoods}, $z$ is also a neighbor of $u$. Thus, $N_G(x) \cap S_1(x) \subseteq N_G(u) \cap S_1(x)$. 

    Since $u \in C_1(x)$, the comp-inclusion property implies that $N_G(x) \cap S_1(x)$ is not a strict subset of $N_G(u) \cap S_1(x)$. Therefore, $N_G(x) \cap S_1(x)$ equals $N_G(u) \cap S_1(x)$.
\end{claimproof}

Since $u \in C_1(x)$, it holds that $x \prec_{\sigma_1} u$. Combining the comp-tie property of $\A_1$ and \cref{claim:chordal-plus-comp4}, we can deduce that $x \prec_\rho u$. Due to \cref{claim:chordal-plus-comp3}, it holds that $u \prec_{\sigma_2} x$. Therefore, $\A_2$ has chosen $u$ before $x$ although both vertices are in the same component of $G - S_2(u)$ and $x$ is to the left of $u$ in tie-breaker $\rho$. The comp-inclusion property of $\A_2$ implies that the neighborhood of $u$ in $S_2(u)$ was not a strict subset of the neighborhood of $x$ in $S_2(u)$. The comp-tie property implies that these neighborhoods are not identical. Therefore, there is a vertex $v \in N_G(u) \setminus N_G(x)$ with $v \prec_{\sigma_2} u$. As $x$ is to the left of $u$ in $\sigma_1$,  \cref{lemma:chordal_neighborhoods} implies that vertex $v$ is to the right of $x$ in $\sigma_1$. Since $u \in C_1(x)$ and $uv \in E(G)$, it holds that $v \in C_1(x)$. However, $v \notin C_2(x)$ since $v \prec_{\sigma_2} x$. Thus, $v \in C_1(x) \setminus C_2(x)$. Since $v \prec_{\sigma_2} u$, this contradicts the choice of $u$. Therefore, $u$ cannot exist and $C_1(x) \setminus C_2(x)$ is empty. This contradicts \cref{claim:chordal-plus-comp2}.
\end{proof}

Now it is easy to show that the forward neighborhood of a vertex is always the same and is ordered in the same way if we use strictly comp-monotone searches with the same tie-breaker on a chordal graph.

\begin{theorem}\label{thm:chordal-plus}
Let $G$ be a chordal graph and $\rho$ be a vertex ordering of $G$. Let $\A_1$ and $\A_2$ be strictly comp-monotone graph searches and let $\sigma_1 = \A_1^+(G,\rho)$ and $\sigma_2 = \A_2^+(G,\rho)$. Then for every vertex $v \in V(G)$ it holds that:
\[N_1(v) := \{w \in N_G(v) \mid w \prec_{\sigma_1} v\} = \{w \in N_G(v) \mid w \prec_{\sigma_2} v\} =: N_2(v). \]
Furthermore, the vertices of $N_i(v)$ are ordered identical in $\sigma_1$ and $\sigma_2$.
\end{theorem}

\begin{proof}
Let $C_1(v)$ and $C_2(v)$ be defined as in \cref{lemma:chordal-plus-comp}. Any neighbor of $v$ that is not to the left of $v$ in $\sigma_i$ is in $C_i(v)$. By \cref{lemma:chordal-plus-comp}, it holds that $C_1(v) = C_2(v)$ and, thus, $N_1(v) = N_2(v)$ for every $v \in V(G)$. By \cref{corol:search-peo}, $N_i(v)$ is a clique. Combining this with \cref{lemma:chordal-plus-comp} shows that the vertices of $N_i(v)$ are ordered identical in $\sigma_1$ and $\sigma_2$.
\end{proof}

The following result is a direct consequence of this theorem.

\begin{theorem}\label{theo:ltrees_chordal}
Let $\A_1$ and $\A_2$ be two strictly comp-monotone graph searches and let $\rho$ be a vertex ordering of $G$. Let $\sigma_1 = \A_1^+(G,\rho)$ and $\sigma_2 = \A_2^+(G,\rho)$. Then the sets $\G(G, \sigma_1)$ and $\G(G, \sigma_2)$ are equal.
\end{theorem}

\begin{proof}
Let $T \in \G(G, \sigma_1)$. Let $w$ be an arbitrary vertex of $G$ and let $v$ be the parent of $w$ in $T$. Then $v$ is to the left of $w$ in $\sigma_1$. Due to \cref{thm:chordal-plus}, $v$ is also to the left of $w$ in $\sigma_2$. Therefore, every parent is to the left of its children in $\sigma_2$. By \cref{lemma:char-g-trees}, $T$ is also in $\G(G, \sigma_2)$. The reverse direction works analogously.
\end{proof}

Thus, we have proven that strictly comp-monotone graph searches have all the same rooted \G-trees on a chordal graph. Combining this with \cref{prop:mcs-chordal} leads to the main result of this section.

\begin{theorem}
    Let $\A$ be a strictly-comp monotone graph search and let $G$ be a graph with $n$~vertices. The rooted \G-tree recognition problem of $\A$ can be solved for $G$
    \begin{enumerate}
        \item in $\O(n^4)$ time if $G$ is a chordal graph, and
        \item in $\O(n^2)$ time if $G$ is a split graph.
    \end{enumerate}
\end{theorem}

\section{Conclusion}

We have introduced the \G-tree recognition problem of graph searches and we were able to prove that it is \NP-complete for all common searches except for Generic Search where it is trivial. Furthermore, we studied these problem for several graph classes. Some interesting cases that we have to leave open are the BFS case for chordal graphs and the cases of LDFS, MCS, and MNS for (chordal) bipartite graphs.

One further direction is the consideration of parameterized algorithms. Beisegel et al.~\cite{beisegel2024graph} have shown that the \cl-tree recognition problem of Generic Search parameterized by the number of leaves of the given spanning tree can be solved in \XP{} time and is \W-hard. It is shown in~\cite{scheffler2025partial} that the Partial Search Order problem of MCS and MNS can be solved in \XP{} time when parameterized by the width of the partial order. This implies the same result for the \G-tree recognition problem of these searches when parameterized by the number of leaves. However, it remains open whether these algorithmic results can be improved to \FPT{} time. Furthermore, there are no results for the \cf-tree, \cl-tree, or \G-tree recognition problems of BFS and DFS when parameterized by the number of leaves, except from those that follow directly from the polynomial-time solvability of the problem.

Another open question concerns the self-similar searches introduced in \cref{sec:g-trees}. In \cref{thm:self-similar-label}, we presented sufficient conditions on the partial order $\prec_\A$ such that LabelSearch$(\prec_\A)$ becomes self-similar. As observed these conditions are not characterizing as MNS is self-similar but does not fulfill the third condition. We leave open whether there is a nice characterization of self-similar instances of LabelSearch.

\bibliographystyle{plainurl}
\bibliography{lit.bib}

\end{document}